%% file: main.tex
\documentclass[11pt,english]{article}

\usepackage[margin=1in]{geometry}

\usepackage{amsmath,amsthm,mathtools,amssymb}
\usepackage{thmtools,thm-restate}
\usepackage{algorithm}
\usepackage{algorithmic}
\usepackage[usenames,svgnames,xcdraw,table]{xcolor}
\definecolor{DarkBlue}{rgb}{0.1,0.1,0.5}
\definecolor{DarkGreen}{rgb}{0.1,0.5,0.1}
\usepackage{hyperref}
\hypersetup{
   colorlinks   = true,
 linkcolor    = DarkBlue, % color of internal links
 urlcolor     = DarkBlue, % color of external links
	 citecolor    = DarkGreen % color of links to bibliography
}
\usepackage[capitalize]{cleveref}
\usepackage{graphicx}
\usepackage{svg}
\usepackage{float}
\usepackage{verbatim}

\include{preamble}

\title{\bfseries Guaranteeing Envy-Freeness under Generalized Assignment Constraints}
\author{Siddharth Barman\thanks{Indian Institute of Science. {\tt barman@iisc.ac.in}} \quad Arindam Khan\thanks{Indian Institute of Science. {\tt arindamkhan@iisc.ac.in}} \quad Sudarshan Shyam\thanks{Aarhus University. {\tt shyam@cs.au.dk}} \quad K. V. N.~Sreenivas\thanks{Indian Institute of Science. {\tt venkatanaga@iisc.ac.in}}}
\date{\empty}
\begin{document}

\maketitle
\input{abstract}
\input{intro}

\input{notation}
\input{divisible-ef.tex}
\input{indivisible-efx.tex}

\input{knapsack-indivisible-efx.tex}
\input{approx-indivisible-efx.tex}
\input{conclusion}

\bibliographystyle{alpha}
\bibliography{references}

\newpage 
\appendix
\input{appendix-divisible-fef}
\input{appendix-indivisible-fefx}

\input{appendix-NP-hardness}
\input{appendix-MNW}
\end{document}

%% file: preamble.tex
\newtheorem{theorem}{Theorem}
\newtheorem{lemma}{Lemma}
\newtheorem{claim}[lemma]{Claim}
\newtheorem{remark}{Remark}
\newtheorem{proposition}[lemma]{Proposition}

\newtheorem{definition}{Definition}

\let\originalleft\left
\let\originalright\right
\renewcommand{\left}{\mathopen{}\mathclose\bgroup\originalleft}
\renewcommand{\right}{\aftergroup\egroup\originalright}
\renewcommand{\epsilon}{\varepsilon}
\newcommand{\eps}{\varepsilon}
\renewcommand{\backslash}{\setminus}
\newcommand*{\textcal}[1]{%
  \textit{\fontfamily{qzc}\selectfont#1}%
}

\newcommand{\dens}{\rho}
\newcommand{\prof}{v}
\newcommand{\size}{s}
\newcommand{\boundvec}{\tau}

\newcommand{\Comment}[1]{\textcolor{black!80}{\texttt{// #1}}}

\newcommand{\Th}{^{\mathrm{th}}}
\newcommand{\supscript}[1]{^{(#1)}}

\newcommand{\angParen}[1]{\left\langle{#1}\right\rangle}

\newcommand{\EF}{\textsf{EF}}
\newcommand{\FEF}{\textsf{FEF}}

\newcommand{\EFx}{\textsf{EFx}}
\newcommand{\FEFx}{\textsf{FEFx}}
\newcommand{\approxFEFx}[1]{{#1}\textsf{-FEFx}}

\newcommand{\DivisibleEF}{\textsc{DivisibleFEF}}
\newcommand{\ComputeFEFx}{\textsc{ComputeFEFx}}
\newcommand{\FindFEFx}{\textsc{FindFEFx}}
\newcommand{\ComputeApproxFEFx}[1]{\textsc{ApxFEFx}_{#1}}

\newcommand{\Knapsack}{\mathtt{Kns}}
\newcommand{\FPTASKnapsack}{\mathtt{ApxKns}}
\newcommand{\FindMinimalEnviedSubset}{\mathtt{FindMinimalEnviedSubset}}
\newcommand{\ApproxMinEnviedSubset}{\mathtt{ApxMinEnvied}}

\newcommand{\Opt}{\textrm{Opt}}

\newcommand{\supp}{\textrm{supp}}
\newcommand{\charity}{\textrm{charity}}
\newcommand{\nphard}{\textrm{NP-Hard}}
\newcommand{\nphardparam}{\mu}

\DeclareMathOperator*{\argmax}{arg\,max}

%% file: abstract.tex
\begin{abstract}
We study fair division of goods under the broad class of generalized assignment constraints. In this constraint framework, the sizes and values of the goods are agent-specific, and one needs to allocate the goods among the agents fairly while further ensuring that each agent receives a bundle of total size at most the corresponding budget of the agent. Since, in such a constraint setting, it may not always be feasible to partition all the goods among the agents, we conform---as in recent works---to the construct of charity to designate the set of unassigned goods. For this allocation framework, we obtain existential and computational guarantees for envy-free (appropriately defined) allocation of divisible and indivisible goods, respectively, among agents with individual, additive valuations for the goods. 

We deem allocations to be fair by evaluating envy only with respect to feasible subsets. In particular, an allocation is said to be feasibly envy-free ($\FEF$) iff each agent prefers its bundle over every (budget) feasible subset within any other agent's bundle (and within the charity). The current work establishes that, for divisible goods, $\FEF$ allocations are guaranteed to exist and can be computed efficiently under generalized assignment constraints. Note that, in the presence of generalized assignment constraints, even the existence of such fair allocations of divisible goods is nonobvious, a priori. Our existential and computational guarantee for $\FEF$ allocations is built upon an incongruity property satisfied across a family of linear programs. This novel proof template is interesting in its own right. 

In the context of indivisible goods, $\FEF$ allocations do not necessarily exist, and hence, we consider the fairness notion of feasible envy-freeness up to any good (\FEFx). Under this notion, an allocation of indivisible goods is declared to be fair iff for each pair of agents, $a$ and $b$, envy-freeness holds for agent $a$ against every feasible and strict subset of $b$'s bundle; a similar guarantee is required with respect to the charity. We show that, under generalized assignment constraints, an $\FEFx$ allocation of indivisible goods always exists. In fact, our $\FEFx$ result resolves open problems posed in prior works, which provide existence guarantees under weaker fairness notions and more specialized constraints. Further, for indivisible goods and under generalized assignment constraints, we provide a pseudo-polynomial time algorithm for computing $\FEFx$ allocations, and a fully polynomial-time approximation scheme (FPTAS) for computing approximate $\FEFx$ allocations. 
\end{abstract}

%% file: intro.tex
\section{Introduction}
A significant body of research---at the interface of mathematical economics and computer science---addresses fairness in resource allocation settings \cite{moulin2004fair,handbook2016}. This growing literature captures various real-world application domains, e.g., fair division of land \cite{shtechman2021fair}, public housing units \cite{deng2013story,benabbou2020finding}, electricity \cite{baghel2022fair}, courses among students \cite{Budish2017CourseMA}, and food donations \cite{aleksandrov2015online}. The fair division literature is typically categorized based on the nature of the underlying resources. In particular, resources that can be fractionally assigned (referred to as divisible goods) have been the focus of classic fair division results; see, e.g., \cite{brams1996fair} and \cite{varian1974equity}. The complementary case of indivisible goods (which have to be integrally assigned) has received more attention in recent years \cite{amanatidis2022fair}. Note that divisible goods capture resources such as land and processing time on machines, and indivisible ones provide a framework for discrete resources, like housing units and inheritance. 

While fairly allocating heterogeneous resources, divisible or indivisible, we are often required to respect allocation constraints. Indeed, in many settings, not all allocations of the goods among the participating agents are feasible. For instance, in the context of land division, a typical requirement is to assign each agent a single, connected plot \cite{brams1996fair}. Motivated by such considerations, a budding thread of research in fair division focuses on fair allocations that further satisfy relevant constraints; see \cite{suksompong2021constraints} for a survey on constraints in fair division. Contributing to this active line of work, the current paper extends the reach of fair division guarantees to a broad class of constraints, namely to generalized assignments constraints. 

In the generalized assignment constraints framework, goods have agent-specific values and sizes, and one needs to allocate the goods among the agents such that each agent receives a bundle of total size at most the corresponding budget of the agent. Note that, in a constraint setting, it may not always be feasible to partition all the goods among the agents. Hence, as in recent works (see, e.g., \cite{chaudhury2021little,caragiannis2019envy,wu2021budget}), we conform to the construct of charity to designate the set of unassigned goods.

Under generalized assignment constraints, the problem of maximizing agents' social welfare (without fairness considerations) is referred to as the Generalized Assignment Problem (GAP). This optimization problem has been extensively studied in combinatorial optimization, approximation algorithms, and operations research \cite{martello1990knapsack}. GAP captures many prominent problems, such as AdWords \cite{mehta2007adwords}, display ads problem \cite{feldman2009online}, the (multiple) knapsack problem \cite{chekuri2005polynomial}, and weighted bipartite matching \cite{khuller1994line}. These instantiations highlight the encompassing nature of generalized assignment constraints and their significance in various application domains; additional applications of GAP, in particular, are provided in the survey article \cite{oncan2007survey}. % and \cite{cattrysse1992survey}. 

Complementing the utilitarian objective of GAP with a focus on fairness, we study the allocation of goods---divisible and indivisible, respectively---under generalized assignment constraints. Our fairness guarantees are in terms of envy-freeness. Under this quintessential notion, an allocation is deemed to be fair (envy-free) iff every agent values the bundle assigned to her over that of any other agent. Note that, in the presence of agent-specific constraints, the bundle assigned to an agent $b$ might not be feasible for another agent $a$. Hence, for fair division under constraints, it is not justified to evaluate envy by considering the value that agent $a$ has for another agent $b$'s entire bundle. Addressing this issue, prior works (see, e.g., \cite{dror,wu2021budget}) adapt the notion of envy-freeness (to settings with constraints) by evaluating envy only with respect to feasible subsets. 

Specifically, an allocation is said to be feasibly envy-free ($\FEF$) iff each agent prefers its bundle over every budget-feasible subset\footnote{Here, in the context of divisible goods and for ease of exposition, we use the term subset (of goods) to denote fractional assignments of the goods.} within any other agent's bundle;  a similar guarantee is required with respect to the charity.  The current work establishes that, for divisible goods, $\FEF$ allocations are guaranteed to exist and can be computed efficiently under generalized assignment constraints (Theorem \ref{thm:divisible-ef}). It is relevant to note that, before the current work, even the existence of $\FEF$ allocations of divisible goods under generalized assignment constraints was not known. Indeed, $\FEF$ is a refinement of envy-freeness and not a restriction. In particular, existence of $\FEF$ allocations is not implied by the fact that, in the absence of constraints, envy-free allocations exist: In the unconstrained setting, one can directly obtain an envy-free allocation of divisible goods by dividing each good equally among the agents. However, such a uniform allocation might not be feasible in the presence of generalized assignment constraints. Furthermore, in the unconstrained setting, allocations of divisible goods that maximize Nash welfare are known to be envy-free \cite{varian1974equity}. Even this standard approach (i.e., maximizing Nash welfare subject to the constraints) fails to provide $\FEF$ allocations of divisible goods under the current constraint setup; see Appendix \ref{sec:mnw-doesnt-work}.

In the context of indivisible goods, $\FEF$ allocations do not necessarily exist,\footnote{Consider a single indivisible good and two identical agents, for whom the good has unit value and zero size. In this fair division instance, no allocation (even the one in which the good it given to charity) is \FEF.} and hence, we consider the fairness notion of feasible envy-freeness up to any good (\FEFx). Under this notion, an allocation of indivisible goods is declared to be fair iff for each pair of agents, $a$ and $b$, envy-freeness holds for agent $a$ against every feasible and \emph{strict} subset of $b$'s bundle; a similar guarantee is required with respect to the charity. In other words, $\FEFx$ mandates that, for each pair of agents, $a$ and $b$, after the removal of \emph{any} good $g$ from agent $b$'s bundle, say $A_b$, envy-freeness holds for agent $a$ with respect to all feasible subsets $S \subseteq A_b\setminus \{g\}$. We note that $\FEFx$ is a direct adaptation---to the constraint setting---of envy-freeness up to any good ($\EFx$). This notion has received significant attention in recent works on fair division of indivisible goods, see, e.g., \cite{chaudhury2021little,caragiannis2019unreasonable}. In particular, the existence of $\EFx$ allocations in the unconstrained setting (and notably without charity) is a central open problem in discrete fair division \cite{procaccia2020technical}. For indivisible goods, $\EFx$ provides a persuasive analog of envy-freeness. Similarly, $\FEFx$ renders a strong fairness guarantee in constraint settings. {To appreciate the high benchmark set by $\FEFx$, note that if in an $\FEFx$ allocation, an agent $b$ is assigned a good $g$ which, by itself, is infeasible for agent $a$, then we must have envy-freeness for $a$ against every feasible subset within $b$'s bundle.}

We show that, under generalized assignment constraints and with charity, an $\FEFx$ allocation of indivisible goods always exists (Theorem \ref{thm:indivisible_efx}). The existential guarantee obtained in the current work strengthens the ones provided in \cite{wu2021budget}, \cite{gan2021approximately}, and \cite{barman2022finding}. The strengthening here is in the following two senses: (i) the prior works address fairness notions which are implied by $\FEFx$, and (ii) the works consider settings in which the size of each good is the same for all the agents, though the budgets can be agent-specific. Notably, the obtained $\FEFx$ existential guarantee (Theorem \ref{thm:indivisible_efx}) positively resolves open problems posed in these prior works; see, e.g., Open Problem 7.2 in \cite{suksompong2021constraints}.

Furthermore, we provide a pseudo-polynomial time algorithm for computing $\FEFx$ allocations of indivisible goods under generalized assignment constraints (Theorem \ref{theorem:pseudo-poly}). Complementing this algorithmic result, we show that, in the current context, computing an $\FEFx$ allocation is {\rm NP}-hard (Theorem \ref{theorem:fefx-np-hard} in Appendix \ref{appendix:fefx-np-hard}). Note that this hardness result for $\FEFx$ stands in contrast to the known polynomial-time algorithms for weaker notions; see \cite{wu2021budget} and \cite{barman2022finding}. Building on the constructive proof of existence for $\FEFx$ allocations, we also provide a fully polynomial-time approximation scheme (FPTAS) for computing approximate $\FEFx$ allocations (Theorem \ref{thm:approx_indivisible_efx}). These algorithmic results for $\FEFx$ hold even under matroid constraints (see Remark \ref{rem:extend-matroid} in Section \ref{sec:approx-fefx}). \\

\noindent
{\bf Our Techniques and Additional Related Work.} We obtain the universal existence and efficient computation of $\FEF$ allocations of divisible goods by developing a property called {\em density domination} (Definition \ref{definition:DD}). At a high level, an allocation of divisible goods is said to satisfy this property iff the fractional assignment of each agent $a$ is supported only on its top-$\tau_a$ most dense goods,\footnote{For an agent $a$, the density of a good $g$ is defined as $g$'s value for $a$ divided by its size for $a$.} for some threshold $\tau_a \in \mathbb{Z}_+$. Furthermore, for each of the $(\tau_a-1)$ most dense goods, the fractional assignment of the good to agent $a$ is at least as much as the good's assignment to any other agent. Also, each of these $(\tau_a-1)$ goods need to be completely allocated among the agents, i.e., no fraction of such a good is left in the charity. Note that this property is defined with respect to thresholds, $\tau_a$, for each agent $a$. Moreover, for any given tuple of thresholds, one can write a polynomially-large linear program to test whether there exists an allocation that upholds density domination. This is in contrast to the $\FEF$ definition, which, in and of itself, does not admit such a succinct verification. We prove that density domination implies envy-freeness (Lemma \ref{lem:density-dom}).   

However, it is a priori nonobvious if there exists any allocation that satisfies density domination, i.e., whether there exist thresholds for which the above-mentioned linear program is feasible. A key technical contribution of the work is to prove that, in fact, a density dominating allocation necessarily exists and the corresponding thresholds can be computed efficiently. This gives us the desired existential and the computational guarantee for $\FEF$ allocations, since density domination implies $\FEF$. Our proof of existence of density dominating allocations is based on a novel incongruity property between the above-mentioned linear program and its relaxation (Lemma \ref{lem:progress}). Our proof template is distinctive in its own right; in particular, it does not invoke a fixed-point theorem or a parity argument. 

Our $\FEFx$ guarantee is obtained by utilizing the idea of swapping minimally envied subsets (Definition \ref{defn:envied_set}) from the charity. This iterative rule was used in \cite{chaudhury2021little} (albeit in the absence of constraints) to find $\EFx$ allocations, with the overarching goal of bounding the size of the charity. When executed in constraint settings, the method from \cite{chaudhury2021little} yields a pseudo-polynomial time algorithm for computing $\FEFx$ allocations under generalized assignment constraints. Furthermore, using the known FPTAS for the standard Knapsack problem, one obtains an FPTAS for computing approximate $\FEFx$ allocations. 

In discrete fair division, another interesting constraint class is obtained via requiring that the bundle assigned to each agent is an independent set of a matroid; see \cite{dror,biswas2018fair,kyropoulou2020almost}. As mentioned previously, under matroid constraints and in the presence of charity, $\FEFx$ computation continues to admit both a pseudo-polynomial time algorithm and an FPTAS.

%% file: notation.tex
\section{Notation and Preliminaries}
\label{sec:notation}
We study the problem of fairly allocating a set of $m \in \mathbb{Z}_+$ goods among a set of $n \in \mathbb{Z}_+$ agents with individual budget constraints. Our work addresses the fair allocation of divisible goods and indivisible goods, respectively. Recall that a divisible good is one that can be assigned fractionally among the agents, whereas an indivisible good has to be integrally allocated.

The cardinal preference of each agent $a \in [n]$, over the goods, is expressed via the valuation function $\prof_a(\cdot)$.  Specifically, we will write $\prof_a(g) \in \mathbb{Q}_+$ to denote the value of a good $g \in [m]$ for an agent $a \in [n]$ and focus on settings in which that agents have additive valuations over the goods. Furthermore, the agent-specific budget constraints are expressed in terms of sizes: write $\size_a(g) \in \mathbb{Q}_+$ to denote the size of any good $g \in [m]$ for each agent $a \in [n]$. Also, with each agent $a \in [n]$, we have an associated budget $B_a \in \mathbb{Q}_+$. Our work addresses generalized assignment constraints,  i.e., we address fair division while conforming to the constraint that each agent $a \in [n]$ is assigned goods with cumulative size (under $\size_a(\cdot)$) at most $B_a$. Below we detail these constraints for the case of indivisible and divisible goods, respectively. Also, via scaling and without loss of generality, we assume that all the values and sizes are integral, $v_a(g), s_a(g) \in \mathbb{Z}_+$, for all agents $a \in [n]$ and goods $g \in [m]$. 

For generalized assignment constraints, an important construct is that of goods' densities. In particular, we will write $\dens_a(g) \coloneqq \frac{\prof_a(g)}{\size_a(g)}$ to denote the density of the good $g \in [m]$ for agent $a \in [n]$.

Throughout, an instance of the fair division problem, under generalized assignment constraints, will be specified as a tuple $\langle [n], [m], \{v_a(g) \}_{a, g}, \{s_a(g)\}_{a,g}, \{B_a\}_a \rangle$. \\

 \noindent
{\bf Indivisible Goods.} In the context of indivisible goods, each agent $a \in [n]$ is assigned a subset of goods, $A_a \subseteq [m]$, which we will refer to as the bundle assigned to agent $a \in [n]$. Recall that the agents have additive valuations and, hence, agent $a$'s value for any subset of indivisible goods $S \subseteq [m]$ satisfies $\prof_a(S) = \sum_{g \in S} \prof_a(g)$.

In the indivisible-goods setting, an allocation $\mathcal{A}=(A_1,\ldots, A_n)$ refers to an $n$-tuple of pairwise disjoint subsets of the goods (i.e., $A_a \cap A_b = \emptyset$ for $a \neq b$), wherein bundle $A_a$ is assigned to agent $a$.

For an agent $a \in [n]$, a subset of goods $S \subseteq [m]$ is said to be feasible iff the total size of $S$, according to agent $a$, is at most $a$'s budget $B_a$, i.e., $s_a(S) = \sum_{g \in S} s_a(g) \leq B_a$. Furthermore, an allocation $\mathcal{A} = (A_1,\ldots, A_n)$ is said to be {feasible} iff for each agent $a$ the assigned bundle, $A_a$, is feasible, $\size_a(A_a) = \sum_{g \in A_a} \size_a(g) \leq B_a$, for all $a \in [n]$. Unless otherwise stated, all the allocations that we encounter are feasible. Hence, for ease of exposition, we will directly use the term allocation instead of feasible allocation.

Note that in a constrained fair division setting, it may not always be feasible to partition all the indivisible goods among the agents. For instance, consider a fair division instance in which the size of each good $g$ with respect to any two agents, $a$ and $b$, is the same ($\size_a(g) = \size_b(g)$), but the combined budget of all the agents is strictly less than the total size of the $m$ goods. In such settings, under any feasible allocation, a subset of the goods must remain unassigned. Conforming to prior works, we utilize the notion of charity of denote the set of unassigned goods. Formally, for any allocation $\mathcal{A}=(A_1, \ldots, A_n)$, the set of goods given to the {charity} is denoted as $C_\mathcal{A} = [m] \setminus \cup_{i=1}^{n} A_i$. When the allocation $\mathcal{A}$ is clear from context, we will drop the subscript and write $C$ to denote the set of goods in charity.

To simplify notation, for any subset of indivisible goods $S \subseteq [m]$ and any good $g \in [m]$, we write $S-g$ to denote $S\setminus\{g\}$ and $S+g$ to denote $S\cup\{g\}$. \\

\noindent
{\bf Divisible Goods.} In the case of divisible goods, we will utilize $m$-dimensional vectors, $x_a = \left(x_{a,1}, x_{a,2}, \ldots, x_{a,m} \right) \in [0,1]^m$, to denote the fractional assignment of the goods to each agent $a \in [n]$; in particular, the $g\Th$ component, $x_{a,g} \in [0,1]$, denotes the fraction of the good $g$ assigned to agent $a$. Here, a (fractional) allocation $x = (x_1, \ldots,x_a, \ldots, x_n) \in [0,1]^{n \times m}$ refers to a tuple of assignment vectors (one for each agent) such that at most one unit of each good is assigned among the agents, i.e., $\sum_{a=1}^n x_{a,g} \leq 1$ for all goods $g \in [m]$.

For any assignment vector $y =(y_1, y_2, \ldots, y_m) \in [0,1]^m$, agent $a$ has value $\prof_a(y) = \sum_{g=1}^m y_g \  \prof_a(g)$ and size $\size_a(y) = \sum_{g=1}^m y_g \ \size_a(g)$. In the divisible goods setting, an allocation $x = (x_1,\ldots, x_n) \in [0,1]^{n \times m}$ is deemed to be feasible iff the fractional assignments uphold the budget constraints of all the agents, i.e., $\size_a(x_a) \leq B_a$ for all agents $a \in [n]$. Analogous to the indivisible goods setting, we will use the construct of charity to denote the unassigned fractions of the goods. In particular, for any allocation $x = (x_1,\ldots, x_n) \in [0,1]^{n \times m}$, write  $x_{\text{charity},g}$ to denote the fraction of good $g$ given to charity, $x_{\text{charity},g} = 1 - \sum_{a=1}^{n} x_{a,g}$.

In the divisible goods setting, for any good $g$ and agent $a$, we continue to denote the density as $\rho_a(g) = {\prof_a(g)}/{\size_a(g)}$.

\noindent
\emph{Vector operations.} Since, in the divisible goods setting we denote an allocation by a tuple of $m$-dimensional vectors, we will define some operations on vectors that will be useful. For any pair of vectors $u, w \in [0,1]^m$, we write $u \leq w$ to denote that $u$ is component-wise dominated by $w$, i.e., $u \leq w$ iff $u_i \leq w_i$ for all components $i \in [m]$. Note that, if $u$ and $w$ are binary vectors (i.e., are characteristic vectors of subsets), then $u \leq w$ corresponds to subset containment.

In addition, for vectors $u, w \in [0,1]^m$, we will write $w-u \in [0,1]^m$ to denote the vector whose $i\Th$ component is equal to $\max \{ 0, w_i-u_i\}$, for all $i \in [m]$. Furthermore, $w+u$ denotes the vector whose $i\Th$ component is equal to $\min\{ 1, w_i + u_i\}$, for all $i \in [m]$. Indeed, if $u$ and $w$ are binary (characteristic) vectors then the vector $w-u$  corresponds to set difference and $w+u$ to set union. Furthermore, the vector $w \cap u$ is defined, component-wise, as $\min\{ w_i, u_i\}$,  for all $i \in [m]$.

The vector $e_i$ denotes the $i\Th$ standard basis vector in $\mathbb{R}^m$. For a vector $u \in [0,1]^m$, the support $\supp(u) \coloneqq \{i \in [m] : u_i >0 \}$.
Hence, $\supp(u)$ denotes the subset of goods that are fractionally assigned under vector $u$. \\

\noindent
{\bf Fairness notions.} Our work obtains universal existential guarantees for two central notions of fairness: (i) envy-freeness (in the divisible goods setting) and (ii) envy-freeness up to any good (in the indivisible goods context).

An allocation is said to be envy-free (\EF) iff every agent $a$ values the bundle assigned to her at least as much as any other agent's bundle. Classic results in fair division literature (see, e.g., \cite{varian1974equity}) show that, in the absence of constraints, an envy-free division of divisible goods is guaranteed to exist. The current work establishes that a  universal existential guarantee holds even under generalized assignment constraints, with a natural adaptation of envy-freeness (defined to accommodate constraints).

Note that, in the presence of agent-specific constraints, the bundle assigned to an agent $b$ might not be feasible for another agent $a$. Hence, for fair division under constraints, it is not justified to evaluate envy by considering the value that agent $a$ has for another agent $b$'s entire bundle. Addressing this issue, in particular, prior works (see, e.g., \cite{dror}) adapt the notion of envy-freeness (to settings with constraints) by evaluating envy only with respect to feasible subsets. That is, an agent $a$ is said to be envious of another agent $b$ (or the charity) only if there exists a subset---within $b$'s bundle (or within charity)---that is both feasible for $a$ and has value more than $a$'s bundle. An allocation wherein no such envy exists is said to be feasibly envy-free (\FEF). Formally,\footnote{Recall that for an allocation $x = (x_1, \ldots,x_a, \ldots, x_n)$, the vector $x_{\text{charity}} \in [0,1]^m$ denotes the unassigned fractions of all the goods, i.e., $x_{\text{charity}, g} = 1 - \sum_{a=1}^n x_{a,g}$, for all $g \in [m]$.}

\begin{definition}[$\FEF$]\label{definition:fef}
In an allocation $x = (x_1,\ldots, x_n) \in [0,1]^{n \times m}$ (of divisible goods), an agent $a \in [n]$ is said to be envy-free towards agent $b \in [n]$ iff for all fractional assignments $y \leq x_b$, with the property that $s_a(y) \leq B_a$, we have $v_a(x_a) \geq v_a(y)$. Similarly, an agent $a \in [n]$ is said to be envy-free towards the charity iff for all $y \leq x_{\text{\rm charity}}$, with the property that $s_a(y) \leq B_a$, we have $v_a(x_a) \geq v_a(y)$.

An allocation $x = (x_1,\ldots, x_n) \in [0,1]^{n \times m}$ is said to be feasibly envy-free (\FEF) iff every agent $a \in [n]$ is envy-free towards all other agents $b \in [n]$ and the charity.
\end{definition}

One can identify $\FEF$ allocations for indivisible goods by instantiating Definition \ref{definition:fef} with binary (characteristic) vectors. However, simple examples rule out the general existence of $\FEF$ allocations in the discrete fair division context. Indeed, even in the absence of constraints, envy-free allocations of indivisible goods are not guaranteed to exist. Motivated, in part, by this consideration, prior works in discrete fair division have considered multiple relaxations of envy-freeness; see \cite{amanatidis2022fair} for a  survey. In this thread of work on discrete fair division, one of the most compelling analogs of envy-freeness is the notion of envy-freeness up to any good (\EFx). Specifically, an allocation $\mathcal{A}= (A_1, \ldots, A_n)$ (of indivisible goods) is said to be $\EFx$ iff for all pairs of agents $a, b \in [n]$, we have $v_a(A_a) \geq v_a(A_b - g)$ for all goods $g \in A_b$. In other words, allocation $\mathcal{A}= (A_1, \ldots, A_n)$ is $\EFx$ iff, for each pair of agents $a, b \in [n]$, agent $a$ is not envious of any
\emph{strict} subset of agent $b$'s bundle.

For generalized assignment constraints and with charity, we next define an adaptation of $\EFx$ -- namely feasibly envy-free up to any good (\FEFx). Such an  adaptation was considered in \cite{dror} for matroid constraints.

\begin{definition}[$\FEFx$]\label{definition:fefx}
In an allocation $\mathcal{A}=(A_1,\ldots,A_n)$ of indivisible goods, an agent $a\in[n]$ is said to be $\FEFx$ towards an agent $b \in [n]$ iff for every strict subset $S \subsetneq A_b$, with the property that $s_a(S) \leq B_a$, we have $ \prof_a(A_a) \geq \prof_a(S)$. Similarly, an agent $a\in[n]$ is said to be $\FEFx$ towards the charity $C = [m] \setminus \cup_{i=1}^{n} A_i$ iff, for every strict and feasible subset $S \subsetneq C$, we have $ \prof_a(A_a) \geq \prof_a(S)$.

An allocation $\mathcal{A}$ is said to be $\FEFx$ iff every agent $a\in[n]$ is $\FEFx$ towards every agent $b\in[n]$ and the charity.
\end{definition}
Note that, by definition, in an $\FEFx$ allocation $\mathcal{A}=(A_1,\ldots,A_n)$, for each pair of agents $a, b \in [n]$ (with $A_b \neq \emptyset$), for any good $g\in A_b$ and, subsequently, for every subset $S \subseteq A_b-g$,  with the property that $s_a(S) \leq B_a$, we have $ \prof_a(A_a) \geq \prof_a(S)$. A similiar guarantee holds for every agent $a \in [n]$ and with respect to the charity.

%% file: divisible-ef.tex
\section{$\FEF$ Allocations of Divisible Goods}
\label{section:divisible-fef}
This section develops a polynomial-time algorithm for finding \FEF{} allocations of divisible goods under generalized assignment constraints. The guaranteed success of the algorithm establishes the universal existence of $\FEF$ allocations under these constraints.

Given any fair division instance, $\langle [n], [m], \{v_a(g) \}_{a, g}, \{s_a(g)\}_{a,g}, \{B_a\}_a \rangle$, for design and analysis purposes, we will include a `fictional' good ${m+1}$ with value $v_a({m+1}) =0$ and size $s_a({m+1}) = 2 n \max_b B_b$, for all agents $a \in [n]$. The inclusion of this good, $m+1$, ensures, in particular, that we can always work with allocations $x=(x_1, \ldots, x_n) \in [0,1]^{n \times (m+1)}$ in which the budget constraint of every agent holds with equality, i.e., $s_a(x_a) = B_a$.\footnote{One can obtain this equality by appropriately setting $x_{a, m+1} \in \left[0, \frac{1}{2n} \right]$.} 

The inclusion of the $(m+1)\Th$ good also implies that, in this section, the fractional assignment to each agent $a$ is denoted by an $(m+1)$-dimensional vector $x_a=(x_{a,1}, x_{a,2},\dots,x_{a,m},x_{a,m+1})$. Here, $x_{a,g}$ denotes the fraction of the good $g$ assigned to $a$. The following proposition notes that the included good $m+1$ has no bearing on feasible envy-freeness (\FEF).  

\begin{restatable}{proposition}{PropNoCons}
\label{proposition:no-cons} 
Let $\langle [n], [m], \{v_a(g) \}_{a, g}, \{s_a(g)\}_{a,g}, \{B_a\}_a \rangle$ be a fair division instance with generalized assignment constraints and let $x \in [0,1]^{n \times (m+1)}$ be an $\FEF$ allocation in the constructed instance, with $m+1$ goods. Then, setting $\overline{x}_{a,g} = x_{a,g}$ for all agents $a \in [n]$ and all goods $g \in [m]$ yields an $\FEF$ allocation $\overline{x} \in [0,1]^{n \times m}$, for the underlying instance with $m$ goods. 
\end{restatable} 
The proof of the proposition is deferred to Appendix \ref{appendix:divisible-fef}.

\subsection{Density Domination implies Feasible Envy-Freeness}
As mentioned previously, considering the densities of the goods provides important insights for achieving envy-freeness under generalized assignment constraints. Recall that for any good $g$ and agent $a$, the density $\dens_a(g) \coloneqq \frac{\prof_a(g)}{\size_a(g)}$. 

In this subsection, we will define the density domination property (Definition \ref{definition:DD}) and prove that any fractional allocation that satisfies this property is feasibly envy-free. The idea of density domination and its connection to envy-freeness are novel contributions of this work. 

To define density domination, we will consider, for each agent $a \in [n]$, the density ordering $\pi_a: [m+1] \mapsto [m+1]$ across the goods. Specifically, $\pi_a(t)$ denotes the $t\Th$ most dense good according to $\dens_a(\cdot)$, for each index $1 \leq t \leq (m+1)$. If two goods have the same density for $a$, we break the ties according to the original indexing for goods. Note that the definition of the density ordering $\pi_a$ ensures that, for each index $1 \leq t \leq m$, exactly one of the following conditions hold 
\begin{itemize}
	\item $\dens_a(\pi_a(t)) > \dens_a(\pi_a(t+1))$.
	\item $\dens_a(\pi_a(t)) = \dens_a(\pi_a(t+1))$ and $\pi_a(t) < \pi_a(t+1)$.
\end{itemize}
Also, for each agent $a \in [n]$, we have $\pi_a(m+1) = m+1$. 

Next, we define sets that, for each agent $a \in [n]$ and threshold $\tau_a \in \mathbb{Z}_+$, denote the $\tau_a$ most dense goods for agent $a$.

\begin{definition}[Internal goods and Edge good] \label{definition:int-ext}
For any integer vector $\tau = (\tau_1, \tau_2, \ldots, \tau_n) \in \mathbb{Z}^n_+$, with $\| \tau \|_\infty \leq m+2$, and for any agent $a \in [n]$, the set of \emph{internal goods}, $I_a(\tau)$, is defined as the set of the $(\tau_a -1)$ most dense goods for agent $a$, i.e., 
\begin{align*}
I_a(\tau) \coloneqq 
\left\{ \pi_a(1),  \pi_a(2), \dots , \pi_a(\tau_a-1) \right\}. 
\end{align*}
In addition, for agent $a$, the \emph{edge good} set $E_a(\tau)$ is defined as 
\begin{align*}
E_a(\tau) \coloneqq 
\begin{cases}
\left\{ \pi_a(\tau_a) \right\} & \text{ if } \tau_a \leq m+1  \\
 \emptyset & \text{ otherwise, if } \tau_a = m +2.  \\
\end{cases}
\end{align*}
In addition, the sets $I(\tau) \coloneqq \cup_{a=1}^{n} I_a(\tau)$ and $E(\tau) \coloneqq \cup_{a=1}^{n} E_a(\tau)$ are called the set of internal and edge goods, respectively.
\end{definition}
Note that in this definition, $\tau_a = m+2$ denotes that, for agent $a \in [n]$, all the goods are internal ($I_a(\tau) = [m+1]$) and the edge set is empty ($E_a(\tau)=\emptyset$). Complementarily, if $\tau_a=1$, then the set of internal goods, $I_a(\tau)$, is empty.

We are now ready to define the density domination property.\footnote{Recall that, by convention, we use the term allocation to refer to a feasible allocation, i.e., one that satisfies the budget constraints of all the agents.}

\begin{definition}[Density Domination]\label{definition:DD}
An allocation $x=(x_1, \ldots, x_n)$ is said to satisfy the density domination property iff 
there exists an integer vector $\widehat{\tau} \in \mathbb{Z}^n_+$ (with $\| \widehat{\tau} \|_\infty \leq m+2$) such that for all agents $a,b \in[n]$ we have 
\begin{align*}
&x_{a,g} \geq x_{b,g} & \quad \text{ for all goods } g \in I_a(\widehat{\tau}), \\
&\sum_{g \in I_a(\widehat{\tau}) \cup E_a(\widehat{\tau})} x_{a,g} \ s_{a}(g)   = B_a,  & \quad \text{ and }\\
&\sum_{a=1}^{n} x_{a,g} = 1 & \quad \text{ for all goods } g \in I(\widehat{\tau}).
\end{align*}
\end{definition}

In this definition, the first set of inequalities assert that, if good $g$ is internal to an agent $a$, then the fraction of $g$ assigned to $a$ is at least as much as the good's fractional assignment to any other agent.  The second equation requires that under the allocation $x$, for every agent $a$, the budget constraint is satisfied with an equality. The final condition mandates that if a good $g$ is internal to any agent, then no fraction of $g$ is left unassigned, i.e., it is entirely divided among all the agents in allocation $x$.

It is also relevant to note that the first set of inequalities in the definition of density domination ensure that, if a good $g$ is internal to two agents $\overline{a}$ and $\overline{b}$, then the fractions of $g$ that $\overline{a}$ and $\overline{b}$ receive must be exactly equal, $x_{\overline{a}, g} = x_{\overline{b},g}$ Further, we note that in an allocation that satisfies the density domination property, each agent $a$ is allocated fractions of only the top $\widehat{\tau}_a$ densest goods according to her, i.e., $I_a(\widehat{\tau}) \subseteq \supp(x_a) \subseteq I_a(\widehat{\tau}) \cup E_a(\widehat{\tau})$. This follows from the second set of constraints in \cref{definition:DD}.
Hence, we also have $s_a(x_a) = B_a$. 

We next establish a crucial result about density domination.
\begin{lemma}
\label{lem:density-dom}
Any allocation $x=(x_1, \ldots, x_n)$ that satisfies the density domination property is \FEF.
\end{lemma}
\begin{proof}
To show that the given allocation $x$ is $\FEF$---i.e., it satisfies Definition \ref{definition:fef}---consider any two agents $a, b \in [n]$ and any fractional assignment $y \leq x_b$ with the property that $s_a(y) \leq B_a$. 

To prove that $v_a(x_a) \geq v_a(y)$, we consider the fractional assignments $(y-x_a) \in [0,1]^{m+1}$ and $(x_a - y) \in [0,1]^{m+1}$; recall the vector operations detailed in Section \ref{sec:notation}. In addition, write $\widehat{\tau}=(\widehat{\tau}_1, \ldots, \widehat{\tau}_n)$ to denote the integer vector that certifies the density domination of $x$. Note that the density domination property implies that, for all internal goods $g \in I_a(\widehat{\tau})$, we have $x_{a,g} \geq x_{b,g} \geq y_g$. Using this inequality and the definition of $I_a(\widehat{\tau})$, we obtain that all the goods in the set $\supp(y-x_a) = \left\{ g' \in [m+1] \ : \ y_{g'} > x_{a, g'} \right\} $ have density \emph{at most} $\dens_a(\pi_a(\widehat{\tau}_a))$. 

As mentioned previously, the density domination property also mandates that all the goods in the set $\supp(x_a)$ have density at least $\dens_a(\pi_a(\widehat{\tau}_a))$. Since $\supp(x_a - y) \subseteq \supp(x_a)$, the density of every good in $\supp(x_a - y)$ is \emph{at least} $\dens_a(\pi_a(\widehat{\tau}_a))$.

Note that $x_a = (x_a - y) + (x_a \cap y)$ and $y = (y - x_a) + (x_a \cap y)$. These equations lead to the following size bound 
\begin{align*}
s_a(x_a - y) + s_a(x_a \cap y) & = s_a(x_a) \\
& = B_a \tag{via density domination}  \\
& \geq s_a(y) \tag{via feasibility of $y$}  \\
& = s_a(y - x_a) + s_a(x_a \cap y).
\end{align*}
The last inequality reduces to $s_a(x_a - y) \geq s_a(y - x_a)$. These observations establish that agent $a$ values its bundle at least as much as the fractional assignment $y$: 
\begin{align*}
v_a(x_a)  & = v_a(x_a \cap y)  + v_a(x_a - y)  \\
& \geq v_a(x_a \cap y)  + \dens_a(\pi_a(\widehat{\tau}_a)) \ s_a(x_a - y) \\
& \geq v_a(x_a \cap y) + \dens_a(\pi_a(\widehat{\tau}_a)) \ s_a(y - x_a) \\
& \geq v_a(x_a \cap y) + v_a(y - x_a) \\
& = v_a(y).
\end{align*}
Therefore, we get that, in the density dominating allocation $x$, every agent $a \in [n]$ is envy-free towards all other agents (see Definition \ref{definition:fef}). 

We now prove that no agent envies the charity. Let $x_{\charity,g}$ denote the fraction of good $g$ allocated to the charity, i.e., $x_{\charity, g} = 1 - \sum_{a=1}^n x_{a, g}$. Fix any agent $a \in [n]$ and any good $\widehat{g} \in I_a(\widehat{\tau})$. Note that the third set of equations in the density domination property (see Definition \ref{definition:DD}) implies that, for all internal goods $\widehat{g} \in I_a(\widehat{\tau})$, we have $x_{a,\widehat{g}} \ge x_{\charity, \widehat{g}}=0$. 
Therefore, by arguments similar to those above, we obtain that $a$ does not envy the charity. The lemma stands proved. 
\end{proof}
It is relevant to note that Lemma \ref{lem:density-dom} and Proposition \ref{proposition:no-cons} imply that from a density dominating allocation $x$ (and by removing the $(m+1)$th good from consideration), one obtains an $\FEF$ allocation for the underlying instance.

\subsection{$\FEF$ Algorithm}
To capture the density domination property and develop our algorithm, we first define linear programs, $LP_1(\cdot)$, that are parameterized by integer vectors $\tau \in \mathbb{Z}^n_+$.
\begin{definition}
Given any integer vector $\tau =(\tau_1, \ldots, \tau_n) \in \mathbb{Z}_+^n$ (with $\| \tau \|_\infty \leq m+2$), we define the following linear program, $LP_1(\tau)$, over decision variables $\left\{z_{a, g} \in [0,1] \right\}_{a, g}$: 
\begin{align*}
&z_{a,g} \geq z_{b,g}  & \quad \text{for all } a,b \in [n] \text{ and } g \in I_a(\tau) \\
&\sum_{g \in I_a(\tau) \cup E_a(\tau)} z_{a, g} \  s_{a}(g) = B_a & \quad \text{for all } a\in [n]\\
&\sum_{a=1}^{n} z_{a,g} = 1 & \quad \text{for all } g \in I(\tau) \\
&z_{a,h} = 0 & \quad \text{for all } a \in [n] \text{ and } h \notin \left( I_a(\tau) \cup E_a(\tau) \right) \\
&\sum_{a=1}^n z_{a,h} \leq 1 & \quad \text{for all } h \in [m+1] \setminus  I(\tau).
\end{align*}
\end{definition}

In addition to the requirements from density domination, the linear program $LP_1(\tau)$ contains a fourth and fifth set of constraints. These additional constraints ensure that the $z_{a,g}$-s induce an allocation. Also, we note that, the fourth set of constraints are redundant in $LP_1(\cdot)$, since they follow from the second set of constraints. Nonetheless, the fourth set of constraints will become relevant in the relaxation mentioned below.

Moreover, as stated in the following proposition, the feasibility of $LP_1(\cdot)$ implies existence of density dominating allocations; the proof of the proposition is direct and, hence, omitted. 

\begin{proposition}\label{proposition:feasible-ef}
If for an integer vector $\tau \in \mathbb{Z}_+^n$ the linear program $LP_1(\tau)$ is feasible, then a feasible solution $\left\{z_{a, g} \in [0,1] \right\}_{a, g}$ corresponds to an allocation $z \in [0,1]^{n \times (m+1)}$ that satisfies the density domination property.  
\end{proposition}

Note that, a priori, it is not clear that $LP_1(\tau)$ is feasible for any integer vector $\tau \in \mathbb{Z}^n_+$. However, if there exists a $\tau$ that induces a feasible $LP_1(\tau)$, then, by \cref{proposition:feasible-ef} and \cref{lem:density-dom}, we will obtain the desired existential guarantee for $\FEF$. The rest of the section is dedicated to showing that such a $\tau$ indeed exists and, moreover, it can be computed in polynomial time.

As a first step, we formulate a new family of linear  programs, $LP_2(\tau)$ (again parameterized by integer vectors), by relaxing the second set of constraints of $LP_1(\tau)$, i.e., we relax the requirement that for every agent the budget constraint holds with an equality.
\begin{definition}
Given any integer vector $\tau =(\tau_1, \ldots, \tau_n) \in \mathbb{Z}_+^n$ (with $\| \tau \|_\infty \leq m+2$), we define the following linear program, $LP_2(\tau)$, over decision variables $\left\{z_{a, g} \in [0,1] \right\}_{a, g}$: 
\begin{align*}
&(C_\textcal{1}) \qquad \qquad  z_{a,g} \geq z_{b,g}  & \quad \text{for all } a,b \in [n] \text{ and } g \in I_a(\tau) \\
&(C_\textcal{2}) \qquad \qquad \sum_{g \in I_a(\tau) \cup E_a(\tau)} z_{a, g} \  s_{a}(g) \leq  B_a & \quad \text{for all } a\in [n]\\
&(C_\textcal{3}) \qquad \qquad \sum_{a=1}^{n} z_{a,g} = 1 & \quad \text{for all } g \in I(\tau) \\
&(C_\textcal{4}) \qquad \qquad z_{a,h} = 0 & \quad \text{for all } a \in [n] \text{ and } h \in [m+1] \setminus \left( I_a(\tau) \cup E_a(\tau) \right) \\
&(C_\textcal{5}) \qquad \qquad \sum_{a=1}^n z_{a,h} \leq 1 & \quad \text{for all } h \in [m+1] \setminus  I(\tau).
\end{align*}
\end{definition}

Now, we prove an important lemma of the section, which establishes an incongruity between the linear programs. 
\begin{lemma}
\label{lem:progress}
Let $\tau \in \mathbb{Z}^n_+$ be an integer vector with $\| \tau \|_\infty \leq m+1$. If linear program $LP_2(\tau)$ is feasible and $LP_1(\tau)$ is infeasible, then there exists an agent $k \in [n]$ such that $LP_2 \left(\tau + e_k \right)$ is feasible.
\end{lemma}
\begin{proof}
For vector $\tau \in \mathbb{Z}^n_+$, given that the program $LP_2(\tau)$ is feasible, we consider among its feasible solutions one that maximizes $\sum_{a=1}^{n} z_{a,\pi_a(\tau_a)}$. Write $z^* = \left\{z^*_{a,g} \right\}_{a,g}$ to denote such a feasible solution, i.e., among all feasible solutions (of  $LP_2(\tau)$), the solution $z^*$ maximizes the fractional assignment across the edge goods $\{ \pi_a(\tau_a) \}_a$. 

In addition, consider an agent $b$ for whom the budget constraint holds with a strict inequality, i.e., for whom $\sum_{g \in I_b(\tau) \cup E_b(\tau)} \  z^*_{b, g} \  s_{b}(g) <  B_b$. Note that such an agent $b$ necessarily exists, otherwise $LP_1(\tau)$ would be feasible. Write $\widehat{g} \coloneqq \pi_b({\tau}_b)$ to denote the edge good of agent $b$, i.e., $E_b(\tau) = \{ \ \widehat{g} \ \}$. 
First, we claim that
\begin{align}
    \sum_{a=1}^{n} z^*_{a, \widehat{g}} = 1\label{eq:bound-good-full}
\end{align}
Say, towards a contradiction, that $\sum_{a=1}^{n} z^*_{a, \widehat{g}} < 1$. This strict inequality and the feasibility of $z^*$ implies that $\widehat{g}$ is not an  internal good for any agent, i.e., $\widehat{g} \notin I(\tau)$. Hence, $\widehat{g}$ does not participate in $C_\textcal{1}$ (the first set of constraints) in $LP_2(\tau)$. In such a case, we can increment $z^*_{b, \widehat{g}}$ by a sufficiently small $\epsilon >0$, while maintaining feasibility and, in particular, satisfying the budget constraint of agent $b$. Such an update, however, increases the objective function $\sum_{a=1}^n z_{a, \pi_a(\tau_a)}$ and, hence, contradicts the fact that $z^*$ maximizes this objective function. Therefore, equation (\ref{eq:bound-good-full}) holds for good $\widehat{g}$. 

Now, let $N$ denote the set of agents that have received a nonzero fraction of good $\widehat{g}$ under $z^*$, i.e., $N \coloneqq \{ a \in [n] \ : \ z^*_{a, \widehat{g}} >0 \}$.  Equation (\ref{eq:bound-good-full}) ensures that $N \neq \emptyset$.  In addition, let $M$ denote the set of agents who have received the maximum fraction of $\widehat{g}$, i.e., $M \coloneqq  \argmax_a z^*_{a, \widehat{g}}$. Note that $M \subseteq N$. 

Also, let $F_e$ and $F_i$ be the set of agents for whom $\widehat{g}$ is an edge good and internal good, respectively, $F_e \coloneqq \left\{ a \in [n] : \widehat{g} \in E_a(\tau)\right\}$ and $F_i \coloneqq \left\{ a \in [n] : \widehat{g} \in I_a(\tau) \right\}$. Indeed, the above-identified agent $b$ is necessarily contained in $ F_e$, and it is possible that there exists an agent $a \in F_e$ with $z^*_{a, \widehat{g}} = 0$, i.e., we can have $(F_e \setminus N) \neq \emptyset$. By contrast, the $C_\textcal{1}$ constraints in $LP_2(\tau)$ ensure that, for every agent $\ell$ for whom $\widehat{g}$ is an internal good, it holds that $z^*_{\ell, \widehat{g}} = \max_{a \in [n]} \ z^*_{a, \widehat{g}} >0$, i.e., 
\begin{align}
F_i \subseteq M \subseteq N \label{eq:containMN}
\end{align}
We obtain another useful containment via the $C_\textcal{4}$ constraints in $LP_2(\tau)$: for $z^*$ and  all agents $a \in N$, either $\widehat{g} \in E_a(\tau)$ or $\widehat{g} \in I_a(\tau)$. Hence, 
\begin{align}
N \subseteq F_i \cup F_e \label{eq:containFie}
\end{align}

Building on the above-mentioned observation, we will next show (in Claim \ref{claim:non-empty} below) that the two sets $F_e$ and $M$ must intersect; the claim essentially follows from the optimality of $z^*$. Using this claim, we will then complete the proof of the lemma. 

\begin{claim}\label{claim:non-empty}
$F_e \cap M \neq \emptyset$. 
\end{claim} 
\begin{proof}
Towards a contradiction, assume that $F_e \cap M = \emptyset$. Then, containments (\ref{eq:containMN}) and (\ref{eq:containFie}) lead to the following equality $M = F_i$. This equality shows that $F_i \neq \emptyset$ and
\begin{align}
z^*_{\ell, \widehat{g}} > z^*_{p, \widehat{g}} & \quad \text{for each $\ell \in F_i$ and any $p \notin F_i$} \label{ineq:gap}
\end{align} 
Relying on this strict inequality, we can select a sufficiently small, but positive, $\epsilon >0$ and update $z^*$ to obtain another feasible solution $z'$ as follows: set $z'_{\ell, \widehat{g}} = z^*_{\ell, \widehat{g}} - \frac{\epsilon}{|F_i|}$, for all agents $\ell \in F_i$, and $z'_{b, \widehat{g}} = z^*_{b, \widehat{g}} + \epsilon$. For all other agent-good pairs, the fractional assignment in $z'$ is the same as in $z^*$.  

We now show that the solution $z'$ is feasible for $LP_2(\tau)$: recall that for agent $b$, under $z^*$, the budget constraint was not tight. Hence, increasing $z^*_{b, \widehat{g}}$ by a sufficiently small $\epsilon$ (as in $z'$) maintains feasibility with respect to the $C_\textcal{2}$ constraints. Good $\widehat{g}$ continues to be fully assigned among the agents, since we have cumulatively reduced the fractional assignments of $\widehat{g}$ among the agents in $F_i$ by $\epsilon$ and increased the assignment to agent $b$ by $\epsilon$. Therefore, the $C_\textcal{3}$ constraints continue to hold for $z'$. Furthermore, the $C_\textcal{4}$ and $C_\textcal{5}$ constraints in $LP_2(\tau)$ also hold for $z'$. Moreover, the strict inequality (\ref{ineq:gap}) implies that there exists an appropriately small $\epsilon >0$ such that even after uniformly decrementing $z^*_{\ell, \widehat{g}}$, for agents $\ell \in F_i = M$, the $C_\textcal{1}$ constraints in $LP_2(\tau)$ are maintained. Therefore, $z'$ satisfies all the constraints in $LP_2(\tau)$. 

Note, however, that the objective function value $\sum_{a=1}^n z_{a, \pi_a(\tau_a)}$ of $z'$ is strictly greater than that of $z^*$; recall that $\pi_b(\tau_b) = \widehat{g}$. This contradicts the optimality of $z^*$. Hence, by way of contradiction, we obtain the stated claim, $F_e \cap M \neq \emptyset$.
\end{proof}

We will now complete the proof of the lemma using Claim \ref{claim:non-empty}. In particular, we will show that, for any agent $k \in F_e \cap M$, the program $LP_2(\boundvec+e_k)$ is feasible; in fact, $z^*$ itself is a feasible solution for $LP_2(\boundvec+e_k)$.

Fix any agent $k \in F_e \cap M$. Since $k \in F_e$, we have $I_k(\tau + e_k) = I_k(\tau) + \widehat{g}$. Additionally, $k \in M = \argmax_a z^*_{a, \widehat{g}}$ and, hence, solution $z^*$ satisfies the $C_\textcal{1}$ constraints in $LP_2(\tau + e_k)$, for agent $k$ and other relevant agents as well. The $C_\textcal{2}$ constraints in $LP_2(\tau + e_k)$ are the same as in $LP_2(\tau)$, hence, $z^*$ continues to be feasible with respect to these budget constraints. In addition, equation (\ref{eq:bound-good-full}) enforces the $C_\textcal{3}$ constraints for good $\widehat{g}$. The $C_\textcal{3}$ constraints hold for all the other goods in $I(\tau + e_k)$ -- this follows from the facts that $I(\tau + e_k) = I(\tau) \cup \{ \widehat{g} \}$ and $z^*$ is a feasible solution with respect to $I(\tau)$. Finally, using the containments that $I_a(\tau+ e_k) \cup E_a(\tau + e_k) \supseteq I_a(\tau) \cup E_a(\tau)$, for all agents $a$, we obtain that the $C_\textcal{4}$ and $C_\textcal{5}$ constraints are satisfied by $z^*$ in $LP_2(\tau+e_k)$ as well. 

Overall, we obtain that $LP_2(\tau + e_k)$ is feasible and the lemma stands proved.
\end{proof}

The following lemma shows that $LP_2(\cdot)$ cannot be incessantly feasible. 

\begin{lemma}
\label{lem:divisible-terminate}
For any integer vector $\tau \in \mathbb{Z}_+^n$, with $\| \tau \|_\infty=m+2$, the program $LP_2(\tau)$ is infeasible.
\end{lemma}
\begin{proof}
Given that, for the given vector $\tau\in \mathbb{Z}_+^n$, one of the components is equal to $m+2$, we have that the good ${m+1} \in I(\boundvec)$; see Definition \ref{definition:int-ext}. 

Now, for $LP_2(\tau)$ to be feasible, the good $m+1$ must be fully assigned among the agents; see the $C_\textcal{3}$ constraints in $LP_2(\tau)$. However, since $\size_a({m+1}) = 2n \ \max_b B_b$, for all agents $a \in [n]$, such an assignment is not possible while maintaining the budget constraints $C_\textcal{2}$ of the agents. The lemma stands proved. 
\end{proof}

With \cref{lem:density-dom,lem:progress,lem:divisible-terminate} in hand, we now state our algorithm for computing \FEF{} allocations in polynomial time.
\begin{algorithm}
\caption{\DivisibleEF %-- Allocate the divisible goods among agents $[n]$ and the charity.
}  \label{algo:fracEF}
\textbf{Input:} Fair division instance $\langle [n], [m], \{v_a(g) \}_{a, g}, \{s_a(g)\}_{a,g}, \{B_a\}_a \rangle$ with divisible goods and generalized assignment constraints.  \\ 
\textbf{Output:} An $\FEF$ allocation. 
\begin{algorithmic}[1]
\STATE Initialize $n$-dimensional integer vector $\boundvec \leftarrow (1,1, \ldots, 1)$.
\WHILE{$LP_1(\boundvec)$ is infeasible}
    \STATE Find $k \in[n]$ such that $LP_2(\boundvec+e_k)$ is feasible. \Comment{Such an agent $k$ always exists (\cref{lem:progress}).}
    \STATE Update $\boundvec\leftarrow \boundvec+e_k$.
\ENDWHILE
\RETURN Allocation $x=(x_1, x_2, \ldots, x_n)$ corresponding to a feasible solution of $LP_1(\boundvec)$.
\end{algorithmic}
\end{algorithm}

\begin{theorem}
\label{thm:divisible-ef}
For any given fair division instance with divisible goods and generalized assignment constraints, \DivisibleEF{} computes an \FEF{} allocation in polynomial time. 
\end{theorem}
\begin{proof}
For the initial integer vector $\tau = (1,\ldots, 1)$, the linear program $LP_2(\boundvec)$ is feasible. This follows from the fact that, for this vector and all agents $a \in [n]$, we have $I_a(\tau) = \emptyset$. Hence, the $C_\textcal{1}$ and $C_\textcal{3}$ constraints are satisfied vacuously. Now, by selecting the all-zeros solution, we can satisfy all the remaining constraints. Hence, at the beginning of the while-loop in the algorithm, the program $LP_2(\tau)$ is feasible, and the algorithm maintains this feasibility as an invariant of the loop. 
 
In particular, at the beginning of each iteration of the while-loop, the program $LP_2(\tau)$ is feasible for the maintained vector $\tau$. Now, if for the current $\tau$, the program $LP_1(\boundvec)$ is feasible, then we return an allocation that is guaranteed to be \FEF{} (\cref{lem:density-dom} and \cref{proposition:feasible-ef}). Otherwise, if $LP_1(\tau)$ is infeasible, then the loop executes and we update $\tau$ to $\tau + e_k$. \cref{lem:progress} guarantees that in the current case we successfully find a $k \in [n]$ such that $LP_2(\tau+ e_k)$ is feasible. Therefore, the desired invariant---feasibility of $LP_2(\cdot)$---is maintained. 

Finally, we note that the while-loop cannot iterate indefinitely. After at most $n(m+1)$ iterations, the maintained vector $\tau$ increments up to satisfy $\| \tau \|_\infty=m+2$. However, by \cref{lem:divisible-terminate}, we know that for such a $\tau$ the program $LP_2(\boundvec)$ is infeasible. These observations imply that the while-loop necessarily terminates in $O(nm)$ iterations and the algorithm returns an $\FEF$ allocation. 

Since each iteration of the while-loop entails solving polynomially-large linear programs, the runtime of the algorithm is polynomially bounded. The theorem stands proved. 
\end{proof}

%% file: indivisible-efx.tex
\section{Existence of $\FEFx$ Allocations}
\label{section:indivisible-fefx}
This section presents a constructive proof of existence of $\FEFx$ allocations under generalized assignment constraints. The proof relies on an algorithm (Algorithm \ref{algorithm:fefx}) that continually finds a \emph{minimal envied subset} $T$ (while one exists) in the charity and swaps $T$ with the bundle of an agent who envies it. We first define the concepts of {envied set} and {minimal envied sets} (Definition \ref{defn:envied_set}), and, subsequently, use them in the algorithm. 

As mentioned previously, in the case of generalized assignment constraints, a subset of indivisible goods $S \subseteq [m]$ is said to be feasible for an agent $a \in [n]$, iff $s_a(S) \leq B_a$. 

\begin{definition}[Envied and Minimal Envied Subsets]
\label{defn:envied_set}
For an allocation $\mathcal{A} = (A_1, \ldots, A_n)$, we say that a set of goods $T \subseteq [m]$ is \emph{envied} by an agent $a \in [n]$ iff there exists a subset $S \subseteq T$ that is feasible for agent $a$ and satisfies $v_a(S) > v_a(A_a)$.

Further, for allocation $\mathcal{A} = (A_1, \ldots, A_n)$, a set of goods $T \subseteq [m]$ is said to be a \emph{minimal envied set} iff the following conditions hold 
\begin{itemize}
\item $T$ is envied by some agent $k \in [n]$.
\item No strict subset $T' \subsetneq T$ is envied by any agent $k' \in [n]$. 
\end{itemize}
\end{definition}

Note that if a set $T$ is envied (by some agent $a \in [n]$), then there necessarily exists $T' \subseteq T$ that is a minimal envied set. 

The algorithm \ComputeFEFx{} (Algorithm \ref{alg:indivisible-efx}) is detailed next. We will show that it finds an $\FEFx$ allocation in finite time. A finite-time termination guarantee for Algorithm \ref{alg:indivisible-efx} suffices for the desired existential guarantee (\cref{thm:indivisible_efx}). The time complexity of the algorithm for generalized assignment constraints and pseudo-polynomial-time implementations of its steps are addressed in Section \ref{section:gac-fefx}. 

\begin{algorithm}
\caption{\ComputeFEFx{}} %Given instance $\langle [m], [n], \{ v_a(g) \}_{a,g}, \{s_a(g)\}_{a,g}, \{B_a\}_{a} \rangle$, compute an allocation that is \FEFx{}}
\textbf{Input:} Fair division instance $\langle [n], [m], \{v_a(g) \}_{a, g}, \{s_a(g)\}_{a,g}, \{B_a\}_a \rangle$ with indivisible goods and generalized assignment constraints.  \\ 
\textbf{Output:} An $\FEFx$ allocation. 
\label{algorithm:fefx}
\begin{algorithmic}[1]
\STATE Initialize allocation $\mathcal{A}=(A_1,\ldots,A_n)=(\emptyset, \ldots,\emptyset)$ and charity $C = [m]$. 
\WHILE{the charity $C$ is envied by any agent $a \in [n]$} \label{line:inner-while}
	\STATE \label{line:line1} Select a minimal envied set $T\subseteq C$ and let $k$ be the agent that envies $T$. 
	\STATE \label{line:line2} Update bundle $A_{k}\leftarrow T$ and charity $C\leftarrow [m]\setminus \left( \cup_{a=1}^n A_a \right)$.
\ENDWHILE\label{line:outer-loop-exit}
\RETURN Allocation $\mathcal{A}$.
\end{algorithmic}
\label{alg:indivisible-efx}
\end{algorithm}

For the purposes of analysis, write $\mathcal{A}\supscript{t}=\left(A_1\supscript{t}, A_2\supscript{t}, \ldots,A_n\supscript{t}\right)$ to denote the allocation maintained by Algorithm \ref{alg:indivisible-efx} just before the $t\Th$ iteration of the while loop of \cref{alg:indivisible-efx}. In particular, $\mathcal{A}\supscript{1} = (\emptyset, \ldots, \emptyset)$. Also, write $C\supscript{t}$ to denote the set of goods in charity just before the $t\Th$ iteration, i.e., $C \supscript{t} = [m] \setminus \left( \cup_{a=1}^n A\supscript{t}_a\right)$. 

Towards establishing that \cref{alg:indivisible-efx} computes an \FEFx{} allocation, we will first show that for any maintained allocation $\mathcal{A}\supscript{t}$, the \FEFx{} property is satisfied among the agents. Then, we will show that when the algorithm ends, the \FEFx{} property (in fact, the stronger $\FEF$ property) is satisfied for every agent against the charity. The proofs of the following two lemmas are direct and delegated to Appendix \ref{appendix:indivisible-fefx}.

\begin{restatable}{lemma}{LemmaEFxAgents}
\label{lem:efx-agents}
For each iteration count $t \geq 1$, the maintained allocation $\mathcal{A}\supscript{t} = \left(A_1\supscript{t}, \ldots,A_n\supscript{t}\right)$ is feasible and it upholds the \FEFx{} criterion among all the agents:  for each pair of agents $a, b \in [n]$ (with $A_b\supscript{t} \neq \emptyset$), and every strict subset $S \subsetneq A_b$, that is feasible for $a$, the following inequality holds $v_a\left(A_a\supscript{t}\right)\ge v_a\left(S\right)$. 
\end{restatable}

The next lemma states that, in Algorithm \ref{alg:indivisible-efx}, if and when the while-loop terminates, every agent bears $\FEF$ against the charity. 

\begin{restatable}{lemma}{LemmaEFxCharity}
\label{lem:efx-charity}
Let the while-loop of Algorithm \ref{alg:indivisible-efx} terminate with allocation $\mathcal{A}=(A_1,\ldots,A_n)$. Then, every agent is $\FEF$ against the charity $C = [m] \setminus \left( \cup_{a=1}^n A_a \right)$, i.e., for each agent $a \in [n]$ and every feasible subset $S \subseteq C$, we have $v_a(A_a) \geq v_a(S)$. 
\end{restatable}

The main result of this section is established next. 

\begin{theorem}
\label{thm:indivisible_efx}
Any fair division instance of indivisible goods with generalized assignment constraints admits an \FEFx{} allocation.
\end{theorem}
\begin{proof}
In each iteration $t > 1$ of while-loop, the algorithm updates the allocation from $\mathcal{A}\supscript{t-1}= \left(A\supscript{t-1}_1, \ldots, A\supscript{t-1}_n \right)$ to $\mathcal{A}\supscript{t}= \left(A\supscript{t}_1, \ldots, A\supscript{t}_n  \right)$. Note that, here, for some agent $k \in [n]$, the value of the assigned bundle strictly increases, $v_k  \left(A\supscript{t-1}_k  \right) > v_k \left(A\supscript{t}_k  \right)$. This follows from the fact that $k$ receives a set $\left( A_k\supscript{t} = T \right)$ that it envies. Furthermore, the bundles of all agents $a \neq k$ remain unchanged. Therefore, the social welfare of the agents strictly increases in each iteration of \ComputeFEFx{}: $\sum_{a=1}^n v_a \left(A\supscript{t}_a  \right) > \sum_{a=1}^n v_a \left(A\supscript{t-1}_a  \right)$. 

Since the social welfare under the initial allocation is zero and the social welfare of any $\mathcal{A}\supscript{t}$ cannot exceed $\sum_{a=1}^n v_a([m])$, we get that the loop terminates in finite time. Furthermore, Lemmas \ref{lem:efx-agents} and \ref{lem:efx-charity} imply that the returned allocation is indeed $\FEFx$. 

The guaranteed success of the algorithm establishes the existence of an $\FEFx$ allocation. The theorem stands proved. 
\end{proof}

%% file: knapsack-indivisible-efx.tex
\subsection{Pseudo-Polynomial Time Algorithm for Finding $\FEFx$ Allocations}
\label{section:gac-fefx}

This section shows that, for generalized assignment constraints, the steps in Algorithm \ref{alg:indivisible-efx} can be implemented such that the algorithm executes in pseudo-polynomial time. Hence, under these constraints, we obtain a pseudo-polynomial time algorithm for finding $\FEFx$ allocations. 

Recall that in the classic Knapsack problem, we are given set of items---each with a weight $w_t \in \mathbb{Q}_+$ and a value $v_t \in \mathbb{Q}_+$---along with a capacity $W \in \mathbb{Q}_+$. The objective here is to find a maximum-valued subset of items with total weight at most $W$. 

Indeed, the problem of determining whether an agent $a \in [n]$ envies a set of goods $T$ (see Definition \ref{defn:envied_set}) corresponds to the Knapsack problem, in which the weights $w_t = s_a(t)$ and values $v_t = v_a(t)$, for all items $t$, along with the capacity $W = B_a$. We will write $\Knapsack(a, T)$ to denote the solution (subset) obtained from such an instantiation of the Knapsack problem. That is, for any agent $a \in [n]$ and set of goods $T \subseteq [m]$, write 
\begin{align}
\Knapsack(a, T) \coloneqq \argmax_{S \subseteq T: s_a(S) \leq B_a} \ \prof_a(S) \label{defn:knapsack}
\end{align}

It is well-known that the Knapsack problem admits a pseudo-polynomial time algorithm; see, e.g., \cite{Kellerer2004}. Hence, for any agent $a$ and any set of goods $T$, the subset $\Knapsack(a, T)$ can be computed in time $O\left(m \ B_a \right)$. {Alternatively, one can compute $\Knapsack(a, T)$ in time $O\left(m \ v_a (T) \right)$. 

Note that, under an allocation $\mathcal{A}=(A_1, \ldots, A_n)$, an agent $a \in [n]$ envies set $T \subseteq [m]$ iff $v_a\left(\Knapsack(a, T)\right) > v_a(A_a)$.  
These observations imply that the execution condition of the while-loop in Algorithm \ref{alg:indivisible-efx} can be implemented in time $O \left(nm \max_a B_a \right)$. 

Next, we detail a subroutine (Algorithm \ref{alg:findminimalsubset}) that provides a pseudo-polynomial implementation of Line \ref{line:line1} of Algorithm \ref{alg:indivisible-efx}.  That is, the subroutine finds a minimal envied set within the charity. 

\begin{algorithm}[H]
\caption{$\FindMinimalEnviedSubset(C,\mathcal{A})$ -- Under allocation $\mathcal{A}$, find a minimal envied set within the charity $C$ and an envying agent $k$.}
\begin{algorithmic}[1] \label{alg:findminimalsubset}
\STATE Initialize set of goods $T=C$ and initialize $k \in [n]$ to be an agent that envies $C$. \Comment{This routine is called only when $C$ is envied by some agent.}
\WHILE{there exists a good ${g'} \in T$ and agent ${a'} \in [n]$ such that $v_{{a'}} \left( \Knapsack\left({a'}, T - {g'} \right) \right) > v_{{a'}} (A_{{a'}})$ (i.e., ${a'}$ envies $(T - {g'})$)}
\STATE Update $T \leftarrow T - g'$ and set agent $k = a'$.  \label{line:update-T} 
\ENDWHILE
\RETURN $(T, k)$
\end{algorithmic}
\end{algorithm}

\begin{restatable}{lemma}{LemmaMinEvSub} \label{lem:ks-min-envied}
The subroutine $\FindMinimalEnviedSubset( C,\mathcal{A})$ correctly computes a minimal envied subset of \ $C$, along with a corresponding envying agent $k$, and it executes in time $O\left({\rm poly}(n, m) \max_a B_a \right)$. 
\end{restatable}

The proof of Lemma \ref{lem:ks-min-envied} appears in Appendix \ref{appendix:knapsack}. Building up on the lemma, the following theorem establishes that $\FEFx$ allocations can be computed in pseudo-polynomial time. 

\begin{theorem}\label{theorem:pseudo-poly}
For any given fair division instance $\langle [m], [n], \{ v_a(g) \}_{a,g}, \{s_a(g)\}_{a,g}, \{B_a\}_{a} \rangle$ with generalized assignment constraints, we can compute an $\FEFx$ allocation in time 
\begin{align*}
O\left( {\rm poly}(n,m) \ \max_{a\in[n]} B_a  \ \max_{a\in[n]} \prof_a([m]) \right).
\end{align*}
\end{theorem}
\begin{proof}
We know that Algorithm \ref{algorithm:fefx} finds an $\FEFx$ allocation (Theorem \ref{thm:indivisible_efx}). Hence, to complete the proof of the theorem it remains to show that the algorithm admits a pseudo-polynomial time implementation. Towards this, first, note that the execution condition of the while-loop (Line \ref{line:inner-while}) in Algorithm \ref{algorithm:fefx} can be evaluated in $O \left(nm \max_a B_a \right)$ time. Recall that testing if an agent $a$ envies the charity $C$ corresponds to verifying whether the following strict inequality holds: $v_a\left(\Knapsack(a, C)\right) > v_a(A_a)$. 

Next, invoking subroutine $\FindMinimalEnviedSubset$, we can execute Line \ref{line:line1} of Algorithm \ref{algorithm:fefx}. Since the time complexity of this subroutine is $O\left({\rm poly}(n, m) \max_a B_a \right)$ (Lemma \ref{lem:ks-min-envied}), even this line of Algorithm \ref{algorithm:fefx} runs in pseudo-polynomial time. Hence, each step of the algorithm can be implemented in pseudo-polynomial time. 

To complete the runtime analysis, we show that the algorithm iterates at most $n \max_{a\in[n]} \prof_a([m])$ times. As observed in the proof of Theorem \ref{thm:indivisible_efx}, the social welfare $\sum_{a=1}^n v_a(A_a)$ strictly increases in each iteration of Algorithm \ref{algorithm:fefx}. Also, recall the assumption that the values and sizes are integral. Hence, in each iteration the welfare increases by at least $1$. Given that the welfare is upper bounded by $n \max_a v_a([m])$, we get that the algorithm iterates at most $n \max_{a\in[n]} \prof_a([m])$ times. The above-mentioned observations show that Algorithm \ref{algorithm:fefx} finds an $\FEFx$ allocation in time $O\left( {\rm poly}(n,m)  \max_{a\in[n]} B_a   \max_{a\in[n]} \prof_a([m]) \right)$. This completes the proof of the theorem. 
\end{proof}

%% file: approx-indivisible-efx.tex
\subsection{FPTAS for \FEFx{} Allocations}
\label{sec:approx-fefx}
This section shows that the pseudo-polynomial time algorithm---detailed in the Section \ref{section:gac-fefx}---can be altered to obtain a polynomial-time algorithm for computing approximate $\FEFx$ allocations. We formally define the notion of approximate $\FEFx$ next. 

\begin{definition}[$\approxFEFx{(1-\eps)}$ allocation] \label{defn:apx-fefx}
For parameter $\eps \in [0,1)$, an allocation $\mathcal{A}=(A_1,\ldots,A_n)$ of indivisible goods is said to be $\approxFEFx{(1-\eps)}$ iff the following two conditions hold
\begin{itemize}
\item For each pair of agents $a, b \in [n]$, and every strict subset $S \subsetneq A_b$, with the property that $s_a(S) \leq B_a$, we have $ \prof_a(A_a) \geq (1-\eps)\prof_a(S)$.
\item Similarly, for each agent $a \in [n]$, and every strict subset $S$ of the charity (i.e., $S \subsetneq C = [m] \setminus \left( \cup_{i=1}^{n} A_i \right)$), with the property that $s_a(S) \leq B_a$, we have $ \prof_a(A_a) \geq (1-\eps)\prof_a(S)$.
\end{itemize}
\end{definition}

We next define the approximate counterparts of {envied sets} and {minimal envied sets}. 

\begin{definition}
\label{defn:approx_envied_set}
For parameter $\eps\in(0,1]$ and allocation $\mathcal{A} = (A_1, \ldots, A_n)$, we say that a set of goods $T \subseteq [m]$ is \emph{$(1-\eps)$-envied} by an agent $a \in [n]$ iff there exists a subset $S \subseteq T$ that is feasible for agent $a$ and satisfies $(1-\eps) v_a(S) > v_a(A_a)$.

Further, for allocation $\mathcal{A}$, a set of goods $T \subseteq [m]$ is said to be a \emph{$(1-\eps)$-minimal envied set} iff the following conditions hold
\begin{itemize}
\item $T$ is $\left(1-\frac{\eps}{2} \right)$-envied by some agent $k \in [n]$, for whom $T$ is feasible as well (i.e., $s_k(T) \leq B_k$).
\item No strict subset $T' \subsetneq T$ is $(1-\eps)$-envied by any agent $k' \in [n]$.
\end{itemize}
\end{definition}

Note that, by definition, a minimal envied set $T$ is feasible for some agent $k$. Also, the envy requirements in the definition of a minimal envied set $T$ are asymmetric; in particular, there exists an agent that $\left(1-\frac{\eps}{2} \right)$-envies $T$ and no agent $(1-\eps)$-envies any strict subset of $T$. 

Recall that, for any agent $a$ and set of goods $T$, we write $\Knapsack(a,T)$ to denote a maximum-valued and feasible (according to $a$) subset $S \subseteq T$; see equation (\ref{defn:knapsack}). Also, under an allocation $\mathcal{A}= (A_1, \ldots, A_n)$, an agent $a$ envies a set $T$ (see Definition \ref{defn:envied_set}) iff $v_a\left( \Knapsack(a,T) \right) > v_a(A_a)$.

To obtain an approximation guarantee for $\FEFx$, we will utilize the known fully polynomial-time approximation scheme (FPTAS) for the Knapsack problem; see, e.g., \cite{Kellerer2004}. Specifically, we write $\FPTASKnapsack(a,T,\eps)$ to denote a feasible subset (solution) obtained by running the FPTAS with input set $T$, knapsack capacity $B_a$, and accuracy parameter $\eps >0$. Hence, the subset $\FPTASKnapsack(a,T,\eps)$ can be computed in time that is polynomial in $\frac{1}{\eps}$ and the input size. In addition, for any $\epsilon \in (0,1]$, we have
\begin{align*}
	\prof_a(\Knapsack(a,T)) \geq v_a\left(\FPTASKnapsack(a,T,\eps) \right) \geq (1-\eps) \prof_a(\Knapsack(a,T)).
\end{align*}

These inequalities directly imply the following proposition.  
\begin{proposition}\label{proposition:kp-envy}
For any allocation $\mathcal{A}=(A_1, \ldots, A_n)$, parameter $\eps \in (0,1/2]$, agent $a \in [n]$, and set $T \subseteq [m]$, if we have $v_a(A_a) < (1-\eps) \ \prof_a (\FPTASKnapsack(a, T, {\eps}))$, then the set $T$ is $(1-\eps)$-envied by agent $a$. Complementarily, the inequality $\prof_a(A_a) \geq \left( 1- {\eps} \right) v_a \left( \FPTASKnapsack \left(a, T, \eps \right) \right)$ implies that the agent $a$ does not $(1-2 \eps)$-envy the set $T$. 
\end{proposition}

We now present the algorithm (Algorithm \ref{alg:approx_indivisible-efx}) to compute approximate $\FEFx{}$ allocations. The algorithm uses $\ApproxMinEnviedSubset$ (Algorithm \ref{alg:approx_findminimalsubset}) as a subroutine. 
\begin{algorithm}
\caption{$\ComputeApproxFEFx{}$}
\textbf{Input:} Fair division instance $\langle [n], [m], \{v_a(g) \}_{a, g}, \{s_a(g)\}_{a,g}, \{B_a\}_a \rangle$ with indivisible goods and generalized assignment constraints along with a parameter $\eps\in(0,1]$.\\
\textbf{Output:} A $\approxFEFx{(1-\eps)}$ allocation.
\begin{algorithmic}[1]
\STATE Initialize allocation $\mathcal{A}=(A_1,\ldots,A_n)=(\emptyset,\ldots,\emptyset)$ and charity $C=[m]$.
\WHILE{there exists an agent $a \in [n]$ such that $\prof_a(A_a)<\left(1- \frac{\eps}{2} \right) \ v_a \left( \FPTASKnapsack \left(a,C, \frac{\eps}{2} \right) \right)$}
	\STATE $({T,k}) = \ApproxMinEnviedSubset(C,\mathcal{A}, \eps)$. \label{line:approx-min-envied}
	%\Comment{A $(1-\eps)$-minimal envied set $T \subseteq C$ and a corresponding $(1-\eps/2)$-envying agent $k$.}
	\STATE Update bundle $A_{k} = T$ and charity $C = [m] \backslash \left(\cup_{a=1}^{n} A_a\right)$.
\ENDWHILE
\RETURN Allocation $\mathcal{A}=(A_1, \ldots, A_n)$.
\end{algorithmic}
\label{alg:approx_indivisible-efx}
\end{algorithm}

\begin{algorithm}
\caption{$\ApproxMinEnviedSubset(C,\mathcal{A}, \eps)$ --
Under allocation $\mathcal{A}$,
find a $(1-\eps)$-minimal envied subset of $C$ and an associated $(1-\eps/2)$-envying agent $k$.}
\begin{algorithmic}[1]
\STATE Initialize set of goods $T=C$. 
\STATE Initialize $k \in [n]$ to some agent for whom $\prof_k(A_k)<\left(1- \frac{\eps}{2} \right) \ v_k \left( \FPTASKnapsack \left(k,T, \frac{\eps}{2} \right) \right)$. \\ \Comment{This strict inequality holds for some agent $k$, whenever the subroutine is called.}\label{line:init-k}
\WHILE{there exists a good $g'\in T$ and agent $a'\in[n]$ such that \\$\prof_{a'}(A_{a'})< \left(1-\eps/2 \right) \ v_{a'} \left( \FPTASKnapsack\left(a', T-g', \eps/2 \right)\right)$}
\STATE Update $T \leftarrow T - g'$ and set agent $k=a'$.\label{line:update-T-approx}
\ENDWHILE
\STATE Update $T \leftarrow \FPTASKnapsack \left(k,T, \frac{\eps}{2} \right)$. \label{line:make-feasible}
\RETURN $({T,k})$
\end{algorithmic}
\label{alg:approx_findminimalsubset}
\end{algorithm}

\begin{restatable}{lemma}{LemmaApxMinEn}
\label{lem:approx-min-envied}
Given any allocation $\mathcal{A}=(A_1, \ldots, A_n)$, any parameter $\eps \in (0,1]$, and charity $C=[m] \setminus \left( \cup_a A_a \right)$, the subroutine $\ApproxMinEnviedSubset(C,\mathcal{A}, \eps)$ (\cref{alg:approx_findminimalsubset}) correctly computes a $(1-\eps)$-minimal envied subset of $C$ and a corresponding $\left(1-\eps/2 \right)$-envying agent $k$. The runtime of the subroutine is polynomial in $\frac{1}{\eps}$ and the input size.
\end{restatable}

%The proof of this lemma is deferred to Appendix \ref{appendix:approx-indivisible-efx}.
\begin{proof}
We first establish the time complexity of the subroutine. As mentioned previously, for any set of goods $S$, any agent $a \in [n]$, and parameter $\eps \in (0,1]$, the subset $\FPTASKnapsack \left(a, S, \eps \right)$ is computed via the standard FPTAS for the Knapsack problem. Hence, in the subroutine (\cref{alg:approx_findminimalsubset}), the initialization and the execution condition of the while-loop  can be implemented in time that is polynomial in $1/\eps$ and the input size. Furthermore, note that the while-loop of the subroutine iterates at most $m$ times, since in each iteration the size of the maintained set $T$ is decremented by $1$. Therefore, as stated in the lemma, runtime of the subroutine is polynomial in $\frac{1}{\eps}$ and the input size. 

Now, we establish the correctness of the subroutine. Throughout the execution of the while-loop in the subroutine, the maintained set $T$ and the corresponding agent $k$ satisfy 
\begin{align}
\prof_k(A_k) & <\left(1- \frac{\eps}{2} \right) \ v_a \left( \FPTASKnapsack \left(k,T, \frac{\eps}{2} \right) \right) \label{ineq:invariant}
\end{align} Indeed, this property is upheld by the selection criterion of the loop and the update in Line \ref{line:update-T-approx}. Hence, the while-loop terminates with a set $T$ and an agent $k$ that satisfy inequality (\ref{ineq:invariant}). In Line \ref{line:make-feasible} of Algorithm \ref{alg:approx_findminimalsubset}, we update $T$ to be the subset $\FPTASKnapsack \left(k,T, \frac{\eps}{2} \right)$. By definition of $\FPTASKnapsack$, we get that this update renders $T$ to be a feasible set of $k$ that additionally satisfies $v_k(A_k) < \left(1- \frac{\eps}{2} \right) v_k(T)$. 

Hence, we obtain that, for returned set $T$, the corresponding agent $k$ satisfies the first condition in the definition of a $(1-\eps)$-minimal envied set; see Definition \ref{defn:approx_envied_set}. 

Therefore, to prove that the returned set $T$ is a $(1-\eps)$-minimal envied set, it remains to show that there does not exist a strict subset $T' \subsetneq T$ and an agent $k' \in [n]$ such that $T'$ is $(1-\eps)$-envied by $k'$ (see Definition \ref{defn:approx_envied_set}). Here, the fact that the while-loop has terminated with set $T$ implies that for all agents $a' \in [n]$ and all goods $g' \in T$ we have $\prof_{a'}(A_{a'}) \geq \left(1-\eps/2 \right) \ v_{a'} \left( \FPTASKnapsack\left(a', T - g', \eps/2 \right)\right)$. Now, using Proposition \ref{proposition:kp-envy}, we get that no agent $(1-\eps)$-envies $T -g'$ for every $g' \in T$. That is, no agent $(1-\eps)$-envies any strict subset of $T$. Line \ref{line:make-feasible} in fact updates $T$ to one of its subsets and, overall, we get that the returned set $T$ is a $(1-\eps)$-minimal envied set. The lemma stands proved.
\end{proof}

Now, we proceed to analyze the algorithm $\ComputeApproxFEFx{}$ (Algorithm \ref{alg:approx_indivisible-efx}). Here, we write $\mathcal{A}\supscript{t}=\left(A_1\supscript{t}, A_2\supscript{t}, \ldots,A_n\supscript{t}\right)$ to denote the allocation maintained by $\ComputeApproxFEFx{}$ just before the $t\Th$ iteration of the while loop of \cref{alg:approx_indivisible-efx}. In particular, $\mathcal{A}\supscript{1} = (\emptyset, \ldots, \emptyset)$.

We first note that the agents' valuations increase multiplicatively as $\ComputeApproxFEFx{}$ progresses. 
\begin{claim}
\label{claim:val-inc-fac}
Let $k \in [n]$ be the agent whose bundle is updated in iteration $t \geq 1$ of $\ComputeApproxFEFx{}$. Then, in the $t\Th$ iteration, agent $k$'s value increases by a multiplicative factor of $\frac{1}{1-\eps/2}$: 
\begin{align*}
v_{k}\left(A_k\supscript{t}\right)>\left( \frac{1}{1-\eps/2} \right) \ v_k\left(A_k\supscript{t-1}\right).
\end{align*}
\end{claim}
\begin{proof}
In Line \ref{line:approx-min-envied} of $\ComputeApproxFEFx{}$ we select a $(1-\eps)$-minimal envied set $T$ and assign it to the corresponding agent $k$ (i.e., set $A\supscript{t}_k = T$). By definition of minimal envied sets, we have that $T$ is $\left(1-\frac{\eps}{2} \right)$-envied by agent $k$. Hence, $v_k \left(A\supscript{t}_k \right) = v_k(T) > \left(\frac{1}{1 - \eps/2} \right) \ v_k \left(A\supscript{t-1}_k\right)$. This completes the proof.  
\end{proof}

The next lemma provides a useful invariant maintained by $\ComputeApproxFEFx{}$. %The proof of the lemma appears in Appendix \ref{appendix:approx-indivisible-efx}. 

\begin{restatable}{lemma}{LemmaApproxEFxAgents}
\label{lem:approx-efx-agents}
In \cref{alg:approx_indivisible-efx} ($\ComputeApproxFEFx{}$), for each iteration count $t \geq 1$, the maintained allocation $\mathcal{A}\supscript{t} = \left(A_1\supscript{t}, \ldots,A_n\supscript{t}\right)$ is feasible and it upholds the $\approxFEFx{(1-\eps)}$ criterion among all the agents: for each pair of agents $a, b \in [n]$ (with $A_b\supscript{t} \neq \emptyset$), and every strict subset $S \subsetneq A_b$, that is feasible for $a$, the following inequality holds $v_a\left(A_a\supscript{t}\right)\geq (1-\eps)v_a\left(S\right)$.
\end{restatable}
\begin{proof}
We provide an inductive proof of the lemma. For the base case, note that  the $\approxFEFx{(1-\eps)}$ property trivially
holds among all the agents before the first iteration, since the initial allocation $\mathcal{A}\supscript{1} = (\emptyset, \ldots, \emptyset)$. Also, note that this allocation is feasible.

For the induction step, consider any iteration count $t>1$. By the induction hypothesis, we know that the allocation $\mathcal{A}\supscript{t-1}$ is feasible and it bears the $\approxFEFx{(1-\eps)}$ property among the agents. In addition, note that in the $t\Th$ iteration, the bundle of exactly one agent gets changed - let that agent be $k \in [n]$. That is, we have $A\supscript{t}_k =T$ and the bundles of all other agents remain unchanged. Lemma \ref{lem:approx-min-envied} implies that, in the iteration, the selected set $T$ (see Line \ref{line:approx-min-envied} of \cref{alg:approx_indivisible-efx}) is a $(1-\eps)$-minimal envied set, with corresponding agent $k$. By definition of minimal envied sets, we get that $T$ is feasible for $k$ and, hence, the feasibility of the allocation  $\mathcal{A}\supscript{t}$ follows. 

We will next complete the induction step by showing that, in the updated allocation $\mathcal{A}\supscript{t}$ and for each agent $a \in [n]$, the $\approxFEFx{(1- \eps)}$ property holds with respect to all other agents $b \in [n]$. Note that for any pair of agents $a, b \in [n]$ one of the following cases apply.

\noindent \emph{Case 1:}  Agent $a\ne k$ and agent $b=k$.
Here, we use the fact that the set $T$ selected in Line \ref{line:line1} is a $(1-\eps)$-minimal envied set; see Lemma \ref{lem:approx-min-envied}. Hence, by definition of minimal envied sets, we get that agent $a$ does not $(1-\eps)$-envy any strict subset of $A\supscript{t}_b = T$. Considering Definitions \ref{defn:approx_envied_set} and \ref{defn:apx-fefx}, we get that, in this case, $a$ is $\approxFEFx{(1-\eps)}$ towards agent $b =k$.

\noindent \emph{Case 2:}  Agent $a= k$ and agent $b\ne k$. In this case, by the induction hypothesis, for every strict subset $S\subsetneq A_b\supscript{t-1} =A_b\supscript{t}$, we have 
\begin{align*}
(1- \eps)\prof_a(S) & \leq v_a\left(A_a\supscript{t-1}\right) \\
& < v_a\left(A_a\supscript{t}\right) \tag{via Claim \ref{claim:val-inc-fac}}
\end{align*}
Hence, in this case as well the desired $\approxFEFx{(1-\eps)}$ guarantee holds.

\noindent \emph{Case 3:} Agent $a \neq k$ and agent $b \neq k$. In this case, $A_a\supscript{t-1}=A_a\supscript{t}$ and $A_b\supscript{t-1}=A_b\supscript{t}$. Hence, by the induction hypothesis, agent $a$ bears $\approxFEFx{(1-\eps)}$ towards $b$.

This exhaustive case analysis completes the proof of the lemma.
\end{proof}

The lemma below states that, in Algorithm \ref{alg:approx_indivisible-efx}, if and when the while-loop terminates, every agent bears $\approxFEFx{(1- \eps)}$ against the charity.

\begin{restatable}{lemma}{LemmaApproxEFxCharity}
\label{lem:approx-efx-charity}
Let the while-loop of $\ComputeApproxFEFx{}$ (\cref{alg:approx_indivisible-efx}) terminate with allocation $\mathcal{A}=(A_1,\ldots,A_n)$.
Then, every agent is $\approxFEFx{(1-\eps)}$ against the charity $C = [m] \setminus \left( \cup_{a=1}^n A_a \right)$; in particular, for each agent $a \in [n]$ and every feasible subset $S \subseteq C$, we have $v_a(A_a) \geq (1-\eps)v_a(S)$.
\end{restatable}
\begin{proof}
The termination condition of the while-loop in $\ComputeApproxFEFx{}$ ensures that, with respect to the allocation $\mathcal{A}$ and for all $a \in [n]$, we have $\prof_a(A_a) \geq \left(1- \frac{\eps}{2} \right) \ v_a \left( \FPTASKnapsack \left(a,C, \frac{\eps}{2} \right) \right)$. Using Proposition \ref{proposition:kp-envy}, we get that no agent $(1-\eps)$-envies the charity. Therefore, every agent is $\approxFEFx{(1-\eps)}$ against the charity $C$; see Definitions \ref{defn:approx_envied_set} and \ref{defn:apx-fefx}. The lemma stands proved. 
\end{proof}

We now establish the main result of this section. 

\begin{theorem}
\label{thm:approx_indivisible_efx}
Given any fair division instance $\angParen{[m], [n], \{ v_a(g) \}_{a,g}, \{s_a(g)\}_{a,g}, \{B_a\}_{a}}$ with generalized assignment constraints and parameter $\eps\in(0,1)$, $\ComputeApproxFEFx{}$ (\cref{alg:approx_indivisible-efx}) computes a $\approxFEFx{(1-\eps)}$ allocation in time polynomial in $\frac{1}{\eps}$ and the input size. 
\end{theorem}
\begin{proof}
Lemmas \ref{lem:approx-efx-agents} and \ref{lem:approx-efx-charity} show that $\ComputeApproxFEFx{}$ computes a $\approxFEFx{(1-\eps)}$ allocation. 

It remains to bound the runtime of the algorithm. Towards this, note that, when the bundle of any agent $a$ gets updated for the first time (i.e., gets updated from the initialized empty bundle), the agent receives a value of at least $\min_g v_a(g)$. Subsequently, in any iteration in which agent $a$'s bundle is updated, $a$'s value strictly increases by a factor of $\delta \coloneqq \frac{1}{1-\eps/2}$ (Claim \ref{claim:val-inc-fac}). Hence, for any agent $a$, the number of iterations
in which its bundle gets updated is upper-bounded by $\log_{\delta}\left(\frac{\prof_a([m])}{\min_{g}\prof_a(g)}\right) + 1$. Summing over all the agents, we obtain that the total number of iterations in $\ComputeApproxFEFx{}$ is upper bounded by a polynomial in $\frac{1}{\eps}$ and the input size (in particular, the bit complexity of the values). Recall that the subset $\FPTASKnapsack(\cdot)$ can be computed via an FPTAS for the Knapsack problem and $\ApproxMinEnviedSubset$ takes time polynomial in $\frac{1}{\eps}$ and the input size (see \cref{{lem:approx-min-envied}}). Combining these arguments, we obtain the stated bound on the time complexity of $\ComputeApproxFEFx{}$. This completes the proof of the theorem.
\end{proof}

\begin{remark}
\label{rem:extend-matroid}
As mentioned previously, discrete fair division under matroid constraints entails assigning each agent $a \in [n]$ a bundle $A_a$ that is independent with respect to a matroid $\mathcal{M}_a = ([m], \mathcal{I}_a)$. Recall that, for any matroid, one can find, in polynomial-time, an independent set with maximum possible value \cite{schrijver}. That is, the matroid-analog of \ $\Knapsack(a, T)$ (see equation (\ref{defn:knapsack})) can be computed in polynomial time. Hence, $\ComputeFEFx$ can be used to design a pseudo-polynomial time algorithm to find $\FEFx$ allocations even under matroid constraints. 

Similarly, $\ComputeApproxFEFx{}$ can be slightly modified to obtain an FPTAS for computing $\FEFx$ allocations under matroid constraints.
\end{remark}

%% file: conclusion.tex
\section{Conclusion and Future Work}
We obtain strong fairness guarantees under generalized assignment constraints. We establish that, under the assignment constraints, feasibly envy-free allocations of divisible goods always exist and can be computed efficiently. Our proof template for this result is interesting in its own right. Furthermore, in the context of indivisible goods and for the assignment constraints, we obtain universal existence of allocations that are feasibly envy-free up to any good. The constructive proof here extends to a pseudo-polynomial time algorithm and an FPTAS. 

Given the positive result for $\FEF$ allocations, it would be interesting to establish an analog of Weller's theorem \cite{weller1985fair}, i.e., to show that, under generalized assignment constraints, there always exists an $\FEF$ allocation of divisible goods that is also Pareto efficient (among all feasible allocations). %Strong impossibility results hold when we include strategy-proofness in the list of desiderata -- Schummer  \cite{schummer1996strategy} has shown that, for divisible goods, the only mechanism that is both strategy-proof and Pareto efficient is serial dictatorship. Hence, one cannot expect fairness in conjunction with economic efficiency and truthfulness. However, an open and interesting direction is to develop a fair and strategy-proof mechanism under generalized assignment constraints and with charity. 
Studying polyhedral constraints---as a generalization of the assignment ones---is another interesting direction of work. 

Under generalized assignment constraints, the computation of $\FEFx$ allocations of indivisible goods is {\rm NP}-hard. However, under matroid constraints, the polynomial-time tractability of $\FEFx$ allocations remains an interesting, open question.

%% file: appendix-divisible-fef.tex
\section{Missing Proof from Section \ref{section:divisible-fef}}
\label{appendix:divisible-fef}

We restate and prove Proposition \ref{proposition:no-cons} here. 

\PropNoCons*
\begin{proof}
The allocation $\overline{x}= \left(\overline{x}_1, \ldots, \overline{x}_n \right)$ is obtained from $x=(x_1, \ldots, x_n)$ by removing the (fictional) good $m+1$ from consideration. Recall that the good $m+1$ has zero value for all the agents. Hence, for all agents $a \in [n]$, the value under the two allocations remains unchanged, $v_a(\overline{x}_a) = v_a(x_a)$. 

Furthermore, we will show that Definition \ref{definition:fef} holds for allocation $\overline{x}$, i.e., $\overline{x}$ is an $\FEF$ allocation in the underling instance. Towards this, fix an agent $a \in [n]$. Now, consider any relevant (feasible) fractional assignment $\overline{y} \in [0,1]^m$ considered for $a$ in the definition; in particular, $\overline{y} \leq \overline{x}_b$, for some agent $b$, or $\overline{y}  \leq \overline{x}_{\charity}$. Extend $\overline{y}$ to obtain the fractional assignment $y  \in [0,1]^{m+1}$ as follows: set $y_g = \overline{y}_g$, for all $g \in [m]$, along with $y_{m+1} =  0$. Note that fractional assignment $y \in [0,1]^{(m+1)}$ is considered when applying Definition \ref{definition:fef} for allocation $x \in [0,1]^{n \times (m+1)}$ and the fact that $x$ is $\FEF$ implies $v_a(x_a) \geq v_a(y)$. Therefore, $v_a(\overline{x}_a) = v_a(x_a) \geq v_a(y) = v_a(\overline{y})$. Since this inequality holds for all agents and relevant fractional assignments $\overline{y} \in [0,1]^m$, we obtain that $\overline{x}$ is an $\FEF$ allocation. 
\end{proof}

%% file: appendix-indivisible-fefx.tex
\section{Missing Proofs from Section \ref{section:indivisible-fefx}}
\label{appendix:indivisible-fefx}
This section restates and proves Lemmas \ref{lem:efx-agents} and \ref{lem:efx-charity}.

\LemmaEFxAgents* 
\begin{proof}
The lemma directly follows from an inductive argument. For the base case, note that just before the first iteration, the \FEFx{} property trivially
holds among all the agents, since we start with $\mathcal{A}\supscript{1} = (\emptyset, \ldots, \emptyset)$. Also, note that this allocation is feasible.

For the induction step, consider any iteration count $t>1$. By the induction hypothesis, we know that the allocation $\mathcal{A}\supscript{t-1}$ is feasible and it bears the \FEFx{} property among the agents. In addition, note that in the $t \Th$ iteration, the bundle of exactly one agent gets changed - let that agent be $k \in [n]$. In this iteration, the selected set $T$ (see Line \ref{line:line1}) is a minimal envied set and is envied by agent $k$. Here, the minimality of $T$ ensures that it is a feasible set for agent $k$ (see Definition \ref{defn:envied_set}). Hence, setting $A\supscript{t}_k =T$, and keeping the bundles of all other agents unchanged, yields a feasible allocation $\mathcal{A}\supscript{t}$. 

We will next complete the induction step by showing that, in the updated allocation $\mathcal{A}\supscript{t}$ and for each agent $a \in [n]$, the $\FEFx$ property holds with respect to all other agents $b \in [n]$. Note that for any pair of agents $a, b \in [n]$ one of the following cases apply. 

\noindent \emph{Case 1:}  Agent $a\ne k$ and agent $b=k$.
Here, we use the fact that the set $T$ selected in Line \ref{line:line1} is a minimal envied set. Assume, towards a contradiction, that, in the updated allocation, agent $a$ is not $\FEFx$ towards agent  $b=k$. That is, there exists a strict subset $S \subsetneq A\supscript{t}_b = T$, that is feasible for $a$ and it satisfies $\prof_a(S)>\prof_a(A_a)$. This, however, contradicts the minimality of $T$, since it admits a strict subset (namely, $S$) that is envied (by agent $a$). Hence, in the current case, $a$ is \FEFx{} towards agent $b =k$. 

\noindent \emph{Case 2:} Agent $a=k$ and agent $b\ne k$. 
An important property of \cref{alg:indivisible-efx} is that whenever the bundle of an agent $a = k$ gets updated in an iteration, its value strictly increases, $v_{a}\left(A_a\supscript{t}\right)>v_a\left(A_a\supscript{t-1}\right)$. This follows from the fact that in Line \ref{line:line1}, we select a set $T$ which is envied by $a=k$ and, hence, has strictly higher value than $A\supscript{t-1}_a$, under valuation $v_a$; see \cref{defn:envied_set}. Then, we swap $A\supscript{t-1}_a$ with $T$. Using this observation and the induction hypothesis, we have $v_a\left(A_a\supscript{t}\right)>v_a\left(A_a\supscript{t-1}\right)\ge v_a(S)$, for every strict and feasible subset $S\subseteq A_b\supscript{t-1} =A_b\supscript{t}$. That is, in this case as well the desired $\FEFx$ guarantee holds. 

\noindent \emph{Case 3:} Agent $a \neq k$ and agent $b \neq k$. In this case, $A_a\supscript{t-1}=A_a\supscript{t}$ and $A_b\supscript{t-1}=A_b\supscript{t}$. Hence, by the induction hypothesis, agent $a$ bears $\FEFx$ towards $b$. 

This exhaustive case analysis completes the proof of the lemma.
\end{proof}

\LemmaEFxCharity*
\begin{proof}
The execution condition of the while loop ensures that the loop terminates only when no agent $a \in [n]$ envies the charity $C$. That is, for every feasible subset $S \subseteq C$ we have $v_a(A_a) \geq v_a(S)$. Hence, the stated $\FEF$ guarantee holds for all agents $a \in [n]$. This completes the proof. 
\end{proof}

\subsection{Missing Proof from Section \ref{section:gac-fefx}}
\label{appendix:knapsack}
Here, we provide a proof of Lemma \ref{lem:ks-min-envied}.

\LemmaMinEvSub*
\begin{proof}
We first establish the time complexity of the subroutine $\FindMinimalEnviedSubset$. Recall that, for any agent $a$ and set of goods $T$, the subset $\Knapsack(a, T)$ can be computed in time $O(m B_a)$. Hence, the execution condition of the while-loop in the subroutine can be evaluated in $O(nm^2 \max_a B_a)$ time. In addition, note that, in each iteration of its while-loop, the subroutine decrements the size of the maintained set $T$. Hence, $\FindMinimalEnviedSubset$ iterates at most $O(m)$ times. Therefore, as stated in the lemma, the time complexity of the subroutine is $O\left({\rm poly}(n, m) \max_a B_a \right)$. 

To complete the proof of the lemma, we show that the returned set is a minimal envied set for the associated agent. Note that throughout the execution of the while-loop in the subroutine, the maintained set $T$ is envied by the corresponding agent $k$ -- this invariant is maintained by update rule (Line \ref{line:update-T}) and selection criterion of the loop. Hence, to prove that the returned set $T$ is a minimal envied set, we need to further show that there does not exist a strict subset $T' \subsetneq T$ and an agent $k' \in [n]$ such that $k'$ envies $T'$; see Definition \ref{defn:envied_set}. However, if such a strict subset $T' \subsetneq T$ exists, then we have a good $g' \in T$ such that $T- g' \supseteq T'$ is envied by $k'$. This would contradict that fact that the while-loop terminated with set $T$ in hand. Therefore, the set $T$ returned by the subroutine is a minimal envied set and the corresponding envying agent is $k$.
The lemma stands proved. 
\end{proof}

%% file: appendix-NP-hardness.tex
\section{Computational Hardness of $\FEFx$ Allocations}
\label{appendix:fefx-np-hard}
This section establishes that finding any $\FEFx$ allocation under generalized assignment constraints is $\nphard$.

\begin{theorem}
\label{theorem:fefx-np-hard}
Given a fair division instance $\langle [n], [m], \{v_a(g) \}_{a, g}, \{s_a(g)\}_{a,g}, \{B_a\}_a \rangle$ with generalized assignment constraints, computing an $\FEFx$ allocation is $\nphard$.
\end{theorem}

\subsection{Proof of Theorem \ref{theorem:fefx-np-hard}}
We provide a polynomial-time reduction from the well-known Knapsack problem to that of finding an $\FEFx$ allocation under generalized assignment constraints.
This Cook reduction establishes the theorem. 

Recall that in the Knapsack problem, we are given a set of items, $[m]$, and a knapsack of capacity $W \in \mathbb{Q}_+$. Each item $j \in [m]$ has a weight $w_j \in \mathbb{Q}_+$ and a value $v_j \in \mathbb{Q}_+$. The objective here is to find a maximum-valued subset of items with total weight at most $W$. Write $\mathcal{I}=\angParen{[m], \{w_i\}_{i\in [m]}, \{v_i\}_{i\in [m]}, W}$ to denote a knapsack instance and let $\Opt(\mathcal{I})$ denote an optimal solution of the instance, i.e., $\Opt(\mathcal{I})$ denotes a maximum-valued subset of weight at most $W$. Also, let $v^*(\mathcal{I})$ denote the optimal value, $v^*(\mathcal{I}) = \sum_{j \in \Opt(\mathcal{I})} \ v_j$.

Towards a Cook reduction, we will show that an algorithm for finding $\FEFx$ allocations can be used as a subroutine to compute, in polynomial time, the optimal value $v^*(\mathcal{I})$ for any given Knapsack instance $\mathcal{I}$. Here, let $\FindFEFx$ denote an algorithm that finds an (arbitrary) $\FEFx$ allocation for any given fair division instance with generalized assignment constraints.

Given a Knapsack instance $\mathcal{I}=\angParen{[m], \{w_i\}_{i\in [m]}, \{v_i\}_{i\in [m]}, W}$, we assume, without loss of generality, that the values of all the items, $v_j$-s, are even integers; otherwise, we can scale the values to be integers and then multiply each by $2$. This assumption ensures that, for the Knapsack instance at hand, the optimal value is an even integer. 

To solve the Knapsack problem with instance $\mathcal{I}=\angParen{[m], \{w_i\}_{i\in [m]}, \{v_i\}_{i\in [m]}, W}$, we will construct fair division instances $\mathcal{F}(\nphardparam)$, parameterized by nonnegative integers $\mu$, and consisting of $(m+2)$ goods and one agent (i.e., $n=1$). We set the budget of the agent $B_1 = W$ and define the sizes and the values of the $(m+2)$ goods for the agent as follows
\begin{align*}
v_1(g)=
		\begin{cases}
			v_g  &  \text{ if }g\in[m]\\
			2\nphardparam+1 &  \text{ if }g=m+1\\
			0 &  \text{ if }g=m+2
		\end{cases}
\qquad
s_1(g)=
		\begin{cases}
			w_g & \text{ if } g\in[m]\\
			W & \text{ if }g=m+1\\
			(W+1) & \text{ if }g=m+2
		\end{cases}
\end{align*}

Note that, for any nonnegative integer $\mu$, in the constructed fair division instance the $(m+1)\Th$ good has an odd value ($2\mu +1$) for the agent. Also, the $(m+2)\Th$ good is infeasible for the agent: $s_1(m+2) = W + 1 > B_1$.

The following lemma will be used in the analysis of the reduction.  

\begin{lemma}
\label{lemma:even-odd}
For any nonnegative integer $\nphardparam$ and constructed fair division instance $\mathcal{F}(\nphardparam)$, let $A_1(\nphardparam)$ be the bundle assigned to the agent in an $\FEFx$ allocation returned by $\FindFEFx$. Then, for all integers $\nphardparam\le \frac{v^*(\mathcal{I})}{2}-1$,  the value $v_1(A_1(\nphardparam))$ is even. On the other hand, for all $\nphardparam\ge \frac{v^*(\mathcal{I})}{2}$, the value $v_1(A_1(\nphardparam))$ is odd.
\end{lemma}
\begin{proof}
As mentioned previously, the good $(m+2)$ cannot be assigned to the agent by $\FindFEFx$, since $s_1(m+2)=W+1>W=B_1$. That is, the good $m+2$ will always be in the charity $C(\nphardparam) \coloneqq [m+2]\setminus A_1(\nphardparam)$.
Hence, for the returned allocation to be $\FEFx$ (see \cref{definition:fefx}), we must have 
\begin{align}
v_1(A_1(\nphardparam))\ge v_1(S) & \quad \text{for any feasible subset }S\subseteq C(\nphardparam)\setminus\{m+2\}  \label{eq:np-hard-cond}
\end{align}
Recall that, by definition, a set $S$ is feasible iff $s_1(S) \leq B_1 = W$. Also, note that the good $m+1$ has size $s_1(m+1) = W=B_1$ and, hence, if this good is assigned to the agent by $\FindFEFx$, then it is the \emph{only} assigned good and all the remaining goods are in charity.

Now, we consider the two specified ranges for the nonnegative integer parameter $\mu$. \\

\noindent
\emph{Case 1:} $\nphardparam\le v^*(\mathcal{I})/2-1$. 
Here, in any $\FEFx$ allocation, the good $m+1$ cannot be assigned to the agent. Otherwise, it must be the case that $\Opt(\mathcal{I})\subseteq C(\nphardparam)\setminus\{m+2\}$. This would contradict \eqref{eq:np-hard-cond}, since we would have a strict, feasible subset $S\subsetneq C(\nphardparam)$ (namely, $S=\Opt(\mathcal{I})$) which the agent envies: $v_1(A_1(\nphardparam))=v_1(m+1)=2\nphardparam+1<2\nphardparam+2\le v^*(\mathcal{I})$.

Hence, we must have $A_1(\nphardparam)\subseteq [m]$. Now, given that the value of each of the first $m$ items is even, we obtain that the value $v_1(A_1(\nphardparam))$ is even.\\

\noindent\emph{Case 2:} $\nphardparam\ge v^*(\mathcal{I})/2$. In this case, in any $\FEFx$ allocation, the good $m+1$ must be allocated to the agent. Otherwise, $\{m+1\}\subsetneq C(\nphardparam)$ which implies that there exists a strict subset in $C(\nphardparam)$ (namely the singleton set $\{m+1\}$)
which the agent envies:  $v_1(A_1(\nphardparam))\le v^*(\mathcal{I})\le2\nphardparam<2\nphardparam+1=v_1(m+1)$.
This is again a contradiction of \eqref{eq:np-hard-cond}. Hence, it must be the case that $A_1(\nphardparam)=\{m+1\}$ and, hence, the value $v_1(A_1(\nphardparam))$ is odd.

This completes the proof of the lemma. 
\end{proof}

Using Lemma \ref{lemma:even-odd}, we infer that for $\nphardparam\in\{0,1,\dots,v^*(\mathcal{I})/2-1\}$, the value $v_1(A_1(\nphardparam))$ is even, and for $\nphardparam\in \left \{v^*(\mathcal{I})/2,\dots,\sum_{j \in [m]} v_j \right \}$, the value $v_1(A_1(\nphardparam))$ is odd. Therefore, via a binary search and, hence, in time $O \left( \log \left( \sum_{j \in [m]} v_j \right) \right)$, we can find the smallest value of $\nphardparam$ for which $v_1(A_1(\nphardparam))$ is odd. Write $\mu^*$ to denote this value; the runtime guarantee of the binary search ensures that $\mu^*$ can be computed in time that is polynomial in the bit complexity of the Knapsack instance $\mathcal{I}$. 

Lemma \ref{lemma:even-odd} implies that, for the given Knapsack instance, the optimal value $v^*(\mathcal{I}) = 2\nphardparam^*$. Hence, using the $\FindFEFx$ algorithm as a subroutine, we can solve the Knapsack problem in polynomial-time. That is, we have a polynomial-time reduction from the {\rm NP}-hard Knapsack problem to that of finding $\FEFx$ allocations. The theorem stands proved. 

%% file: appendix-MNW.tex
\section{Maximum Nash Welfare does not imply Feasible Envy-Freeness}
\label{sec:mnw-doesnt-work}
% The Nash social welfare is defined as the geometric mean of agents' valuations.
The Nash social welfare of a given allocation is defined as the geometric mean of the agents' valuations
of their respective bundles.
Results in the fair division literature highlight the objective of maximizing the Nash social welfare as a means to achieve fairness \cite{caragiannis2019unreasonable,varian1974equity}. This approach, however, does not work in constrained settings.
In particular, via an example below, we show that maximizing Nash social welfare over feasible fractional assignments does not necessarily achieve feasible envy-freeness.

Consider a fair division instance with two agents $\{1,2\}$ and two goods $\{1,2\}$. The agents have equal budgets $B_1 = B_2 = 1$. The sizes and values for the two goods are specified in Figure \ref{figure:size-val}; the left table lists the values and the sizes for the first agent, and the right table is for the second agent.  

\begin{figure}[h]
\centering
\begin{tabular}{|c|c|c|}
\hline
 Good & Value & Size \\
\hline
 1 & 1 & 1  \\
\hline 
 2 & 0.5 & 1 \\
\hline
\end{tabular}
\hspace{20pt}
\begin{tabular}{|c|c|c|}
\hline
 Good & Value & Size \\
\hline
 1 & 1 & 1  \\
\hline 
 2 & 0.5 & 8\\
\hline
\end{tabular}
\caption{The left table contains the values and sizes of the two goods, respectively, for agent $1$. The right table lists these quantities for agent $2$.} 
\label{figure:size-val}
\end{figure}

Write $x^* = (x^*_1, x^*_2)$ to denote the allocation with the following fractional assignments, $x^*_1 = (1/30,29/30)$ and $x^*_2 = (29/30,1/240)$. That is, agent $1$ gets $1/30\Th$ fraction of good $1$ and $29/30\Th$ of good $2$, whereas agent $2$ receives $29/30\Th$ of good $1$ and $1/240\Th$ of good $2$. Note that $x^*=(x^*_1, x^*_2)$ is a feasible allocation: $s_1(x^*_1)=1/30+29/30=1$ and $s_2(x^*_2)=29/30+8/240=1$.

We will show that allocation $x^*$ maximizes Nash social welfare among all feasible allocations (Claim \ref{claim:max-nsw}). At the same time, $x^*$ is not $\FEF$ (Claim \ref{claim:has-envy}). The two claims will establish the stated non-implication. 

\begin{claim}\label{claim:max-nsw}
The allocation $x^*$ has the maximum Nash Social Welfare among all fractional, feasible allocations.
\end{claim}
\begin{proof}
Consider an arbitrary, feasible allocation $x=(x_1, x_2)$, in which $x_{a,g}$ denotes the fraction of the good $g \in \{1, 2\}$ allocated to agent $a \in \{1, 2\}$.
Also, let parameter $\delta \coloneqq 1/8$. Since allocation $x$ is feasible for the given instance, it must satisfy the linear inequalities $K_\textcal{1}$ to $K_\textcal{5}$ stated below:

\begin{align*}
(K_\textcal{1})&\quad x_{1,1} + x_{1,2} \le 1 & \quad \text{(budget constraint of agent $1$)}\\
(K_\textcal{2})&\quad x_{2,1} + \frac{1}{\delta}x_{2,2} \le 1 \quad & \text{(budget constraint of agent $2$)} \\
(K_\textcal{3})&\quad x_{1,1} + x_{2,1} \le 1 & \quad \text{(at most unit fraction of good $1$ allocated to the agents)} \\
(K_\textcal{4})&\quad x_{1,2} + x_{2,2}\le 1 & \quad \text{(at most unit fraction of good $2$ allocated to the agents)}\\
(K_\textcal{5})& \quad x_{1,1}, x_{1,2},x_{2,1},x_{2,2}\in[0,1].
\end{align*}

Under allocation $x$, the product of the agents' values (i.e., square of the Nash welfare) is equal to $(x_{1,1}+0.5x_{1,2})(x_{2,1}+0.5x_{2,2})$. We upper bound this quantity as follows
\begin{align*}
(x_{1,1}+0.5x_{1,2})(x_{2,1}+0.5x_{2,2}) & \leq (x_{1,1}+0.5x_{1,2})\left(x_{2,1}+\frac{\delta}{2}(1-x_{2,1})\right)\tag{via $(K_\textcal{2})$}\\
&=(x_{1,1}+0.5x_{1,2})\left(\frac{\delta}{2}+(1-\delta/2)x_{2,1}\right)\\
& \leq  \frac{1}{2}(1+x_{1,1})\left(\frac{\delta}{2}+(1-\delta/2)x_{2,1}\right)\tag{via $(K_\textcal{1})$}\\
& \leq \frac{1}{2}(1+x_{1,1})\left(\frac{\delta}{2}+(1-\delta/2)(1-x_{1,1})\right)\tag{via $(K_\textcal{3})$}\\
\end{align*}
Define function $f(y) \coloneqq (1+ y)(\delta/2+(1-\delta/2)(1-y))$. Considering the derivative $f'(y)=\delta/2-(2-\delta)y$, one can show that the function $f(y)$ is maximized in $[0,1]$ at $y=\frac{\delta}{2(2-\delta)}$.

Hence, we define $x_{1,1}=\frac{\delta}{2(2-\delta)}$, $x_{1,2}=(1-x_{1,1})$, $x_{2,1}=(1-x_{1,1})$, $x_{2,2}= {\delta}  (1-x_{2,1})$. With this instantiation, all the inequalities $(K_\textcal{1})$ to $(K_\textcal{3})$ become equalities and $(K_\textcal{4})$ along with $(K_\textcal{5})$ are not violated. Furthermore, these choices of fractional assignments insure that all the inequalities in the above-mentioned upper bound hold with equality. That is, under this instantiation, the product of agents' values achieves the upper bound, which is also maximized (via the choice of $x_{1,1}$).

Therefore, this instantiation provides a feasible allocation with maximum Nash social welfare.
Substituting back $\delta=1/8$, we obtain $x_{1,1}=1/30$, $x_{1,2}=29/30$, $x_{2,1}=29/30$, $x_{2,2}=1/240$. These fractional assignments are exactly the same as in allocation $x^*$ and, hence, $x^*$ is a Nash optimal allocation. The claim stands proved.
\end{proof}

\begin{claim}\label{claim:has-envy}
The allocation $x^*$ is not feasibly envy-free ($\FEF$).
\end{claim}
\begin{proof}
Agent $1$ envies agent $2$ in allocation $x^*=(x^*_1, x^*_2)$. Note that, $x^*_2$ is a feasible fractional assignment for agent $1$, since $\size_1(x^*_2)=29/30+ 1/240 < 1 =B_1$. Furthermore, agent $1$'s value for the two bundles satisfy: $\prof_1(x^*_1)=1/30+(1/2)(29/30)=31/60$ and $\prof_1(x^*_2)=29/30+(1/2)(1/240)=465/480$. Hence, $\prof_1(x^*_1)<\prof_1(x^*_2)$, which shows that agent $1$ envies agent $2$. This completes the proof.
\end{proof}

Claim \ref{claim:max-nsw} and \ref{claim:has-envy} establish that a Nash social welfare maximizing allocation need not be feasibly envy-free. 